\documentclass[twoside,journal]{IEEEtran}
\usepackage{amsfonts}
\usepackage{amssymb}
\usepackage{amsmath}
\usepackage{mathrsfs}
\usepackage{verbatim}
\usepackage{subfigure}
\usepackage{balance}
\usepackage{booktabs}
\usepackage{fancybox}
\usepackage{bm}
\usepackage{changepage}
\usepackage{extarrows}
\usepackage{algorithm}
\usepackage{algorithmic}
\usepackage{multirow}
\usepackage{array}
\usepackage{graphicx}
\usepackage{epstopdf}
\usepackage{caption}
\newtheorem{theorem}{Theorem}
\newtheorem{lemma}{Lemma}
\newtheorem{corollary}{Corollary}
\newtheorem{proof}{Proof}
\newtheorem{proposition}{Proposition}

\begin{document}

\title{Safeguarding Decentralized Wireless  Networks Using Full-Duplex Jamming Receivers}
\author{Tong-Xing~Zheng,~\IEEEmembership{Student Member,~IEEE,}
         ~Hui-Ming~Wang,~\IEEEmembership{Senior~Member,~IEEE,~}
         Qian Yang,
        and~Moon~Ho~Lee
\thanks{T.-X. Zheng, H.-M. Wang, and Q. Yang are with the School of Electronic and Information Engineering,
Xi'an Jiaotong University, Xi'an, 710049, Shaanxi, China. Email: {\tt txzheng@stu.xjtu.edu.cn},
{\tt xjbswhm@gmail.com}, {\tt qian-yang@outlook.com}.
}
\thanks{M. H. Lee is with the Division of Electronics Engineering, Chonbuk National 	University, Jeonju 561-756, Korea. Email: {\tt moonho@jbnu.ac.kr}.}
}

\maketitle
\vspace{-0.8 cm}

\begin{abstract}
In this paper, we study the benefits of full-duplex (FD) receiver jamming in enhancing the physical-layer security of a two-tier decentralized wireless network with each tier deployed with a large number of pairs of a single-antenna transmitter and a multi-antenna receiver.
In the underlying tier, the transmitter sends unclassified information, and the receiver works in the half-duplex (HD) mode  receiving the desired signal.
In the overlaid tier, the transmitter deliveries confidential information in the presence of randomly located eavesdroppers, and the receiver works in the FD mode radiating jamming signals to confuse eavesdroppers and receiving the desired signal simultaneously.
We provide a comprehensive performance analysis and network design under a stochastic geometry framework.
Specifically, we consider the scenarios where each FD receiver uses single- and multi-antenna jamming, and analyze the connection probability and the secrecy outage probability of a typical FD receiver by providing accurate expressions and more tractable approximations for the two metrics.
We further determine the optimal deployment of the FD-mode tier in order to maximize network-wide secrecy throughput subject to constraints including the given dual probabilities and the network-wide throughput of the HD-mode tier.
Numerical results are demonstrated to verify our theoretical findings, and show that network-wide secrecy throughput is significantly improved by properly deploying the FD-mode tier.

\end{abstract}


\IEEEpeerreviewmaketitle

\section{Introduction}

\IEEEPARstart{I}{nformation}
security in wireless communications has attracted prominent attention in the era of information explosion.
A traditional approach that safeguards the information security is to use encryption at the upper layers of the communication protocol stack.
However, due to the dynamic and large-scale topologies in emerging wireless networks, secret key management and distribution is difficult to implement, especially in a decentralized network architecture without infrastructure \cite{Poor2012Information}.
In addition, it might not be practical for low-power network nodes, e.g., sensors, to use complicated cryptographic algorithms \cite{Poor2012Information}.
These pose a challenge to securing information delivery solely by means of cryptography-based security mechanisms.
Fortunately, \emph{physical-layer security}, a novel approach at the physical layer that achieves secrecy by exploiting the randomness inherent to wireless channels, has the potential to strengthen network security \cite{Yang2015Safeguading}.
Since Wyner's ground-breaking work \cite{Wyner1975Wire-tap} in which he introduced the \emph{degraded wiretap channel} (DWTC) model and the concept of \emph{secrecy capacity}, physical-layer security has been studied in various wiretap channels models, e.g., multi-input multi-output (MIMO) channels \cite{Liu2010Multiple}, \cite{Khisti2010Secure2},
relay channels \cite{Lai2008Relay}, \cite{Zheng2015Outage}, and two-way channels \cite{Wang2012Distributed}, \cite{Wang2012Hybrid}, etc.
A comprehensive survey on physical-layer security, including the information-theoretic foundations, the evolution of secure transmission strategies, and potential research directions in this area, can be found in \cite{Mukherjee2014Principles}.

Early research on physical-layer security is focused on a point-to-point scenario, in which the large-scale fading is ignored when modeling the wireless channels, and as a consequence secure transmissions become irrelevant to
the relative spatial locations of legitimate terminals and
eavesdroppers.
When it comes to a decentralized wireless network (DWN), since each network node suffers great interference from the other nodes spreading over the entire network, network security strongly depends on nodes' spatial positions and propagation path losses.
Recently, stochastic geometry theory has provided a powerful tool to analyze network performance by modeling nodes' positions according to a spatial distribution, e.g., a Poisson point process (PPP) \cite{Haenggi2009Stochastic}-\cite{Ghogho2011Physical}.
This has facilitated the research of physical-layer security with randomly distributed legitimate nodes and eavesdroppers in DWNs \cite{Vasudevan2010Security}-\cite{Zhang2013Enhancing}.

To improve information transfer secrecy, an efficient way  is to degrade eavesdroppers' decoding ability by sending jamming signals.
Along this line, some efforts have been made.
For example, the authors in \cite{Vasudevan2010Security} and \cite{Zhou2011Throughput} consider a single-antenna transmitter scenario, and propose to let the transmitter suspend its own information delivery and act as a friendly jammer to impair eavesdroppers when it is far away from the intended receiver \cite{Vasudevan2010Security} or when eavesdroppers are detected inside its secrecy guard zone \cite{Zhou2011Throughput}.
The authors in \cite{Zhang2013Enhancing} consider a multi-antenna transmitter scenario, and propose to radiate \emph{artificial noise}\footnote{
	The idea of using artificial noise to interfere with eavesdroppers was first proposed in \cite{Goel2008Guaranteeing};
	this seminal work has unleashed a wave of innovation, mainly including two branches, i.e., multi-antenna techniques \cite{Zhou2010Secure}-\cite{Zheng2016Optimal}
	and cooperative jamming strategies \cite{Wang2015Uncoordinated,Deng2015Secrecy}. } with either sectoring or beamforming to confuse eavesdroppers while without impairing the legitimate receiver.
Although these endeavors are shown to yield a significant improvement on the secrecy capacity/throughput, they are based on the presence of either multi-antenna transmitters or friendly jammers, which sometimes might not be available.
For instance, due to the size and hardware cost constraints, a sensor in a DWN, which transmits sensed data to a data collection station, is usually equipped with only a single antenna.
In addition, a sensor has no extra power to radiate jamming signals due to its low-power constraint.
In these scenarios, the jamming schemes proposed in \cite{Vasudevan2010Security}-\cite{Zhang2013Enhancing} no longer apply, and it is still challenging to protect information from eavesdropping.

Fortunately, the recent progress of developing in-band full-duplex (FD) radios \cite{Sabharwal2014In-band} raises the possibility of enhancing network security in the aforementioned scenarios.
In-band FD operation enables a transceiver to simultaneously transmit and receive on the same frequency band.
The major challenge in implementing such an FD node is the presence of self-interference (SI) that leaks from the node's output to its input.
Nevertheless, thanks to various effective SI cancellation (SIC) techniques, SI can be efficiently mitigated in the analog circuit domain \cite{Duarte2012Experiment}, digital circuit domain \cite{Ng2012Dynamic}, and spatial domain \cite{Riihonen2011Mitigation}, respectively.
FD radios has the potential to improve both link capacity and communication security in DWNs.
Returning to the aforementioned  scenarios, i.e., with single-antenna sensors and no friendly jammer,
using a more powerful FD data collection station provides extra degrees of freedom to protect information delivery, e.g.,
radiating jamming signals to degrade eavesdroppers while receiving desired signals simultaneously.
In particular, when the FD receiver is equipped with multiple antennas, it provides us with potential benefits not only in alleviating SI but also in designing jamming signals.

We point out that sending jamming signals using an FD receiver have already been reported by \cite{Li2012Secure}-\cite{Li2014Secure}, where the authors consider single-antenna receiver jamming with SI perfectly canceled in a cost-free manner, consider multi-antenna receiver jamming with SI taken into account, and consider both transmitter and receiver jamming, respectively.
However, these works are confined to a point-to-point scenario.
When considering a DWN, analyzing the influence of FD radios on network security becomes much more sophisticated due to the presence of not only the mutual interference between nodes but also the SI.
To the best of our knowledge, the potential advantages of FD jamming in the context of physical-layer security from a network perspective are elusive, and a fundamental mathematical framework for performance analysis and network design is lacking, which has motivated our work.

\subsection{Our Work and Contributions}
In this paper, we investigate the physical-layer security of a two-tier heterogeneous DWN under a stochastic geometry framework, where
single-antenna transmitters (sensors) and multi-antenna receivers (data collection stations) in each tier are organized in pairs.
The first tier is an underlying tier that has no secrecy requirement and each receiver therein works in the half-duplex (HD) mode.
The second tier is an overlaid tier that has secrecy considerations and is deployed with more powerful FD receivers.
For convenience, we name the two tiers the HD tier and the FD tier throughout the paper, respectively.
Randomly located multi-antenna eavesdroppers intend to wiretap the secrecy data flowing in the FD tier.
The reasons why we consider this model are:

\begin{itemize}
	
	\item
	This model characterizes a practical communication scenario where a security-oriented network is newly deployed over an existing network that has no security requirement.
	For example, a military ad hoc network specifically for secret information exchange such as offensive tactics, is momentarily added to a civilian ad hoc network, or an unlicensed security secondary tier in an underlay cognitive radio network should make its interference to the primary tier under control to guarantee smooth communications for the latter.
	
	\item
	This is a more general DWN model that incorporates communications with and without security requirements.
	The secure decentralized ad hoc network models discussed in \cite{Zhou2011Throughput} and \cite{Zhang2013Enhancing}
	are just special cases of our model
	when we simply put aside the HD tier.
	
	\item
	In addition, investigating the achievable performances in such a two-tier heterogeneous network facilitates us to gain a better understanding of the interplay between the classified and unclassified networks, and to evaluate the impact of FD jamming  to an existing communication network without security constraint.
\end{itemize}

The main contributions of this paper are summarized as follows:

\begin{itemize}
	
	\item
	We analyze the connection probability and the secrecy outage probability of a typical FD receiver under a spatial SIC (SSIC) strategy, and provide accurate integral expressions as well as analytical approximations for the given metrics.
	We show that deploying more FD nodes introduces greater interference to the network, which not only decreases the connection probability but also decreases the secrecy outage probability.
	
	\item
	We study the optimal deployment of the FD tier to maximize network-wide secrecy throughput subject to constraints including the connection probability, the secrecy outage probability, and the HD tier throughput.
	In particular, when the FD receiver uses single-antenna jamming, we prove the  \emph{quasi-concavity} of the secrecy throughput with respect to (w.r.t.) the FD tier density; the optimal density that maximizes secrecy throughput can be obtained using the bisection method.
	
	\item
	For the multi-antenna jamming scenario, we investigate how the number of jamming signal streams and the number of jamming antennas affect secrecy throughput.
	We reveal that increasing jamming signal streams always benefits secrecy throughput.
	However, whether adding jamming antennas is advantageous or not depends on specific communication environments.
	A proper number of jamming antennas should be chosen to balance transmission reliability with secrecy.
	
\end{itemize}

\subsection{Organization and Notations}
The remainder of this paper is organized as follows.
In Section II, we describe the system model and the underlying optimization problem.
In Sections III and IV, we provide performance analysis and network design in single-antenna jamming and multi-antenna jamming scenarios, respectively.
In Section V, we conclude our work.

\emph{Notations}:
bold uppercase (lowercase) letters denote matrices (vectors).
$(\cdot)^{{H}}$, $|\cdot|$, $\|\cdot\|$, $\mathbb{P}\{\cdot\}$, and $\mathbb{E}_A(\cdot)$ denote Hermitian transpose, absolute value, Euclidean norm, probability, and expectation w.r.t. $A$, respectively.
$\mathcal{CN}(\mu, \nu)$, ${{\mathrm{Exp}}}(\lambda)$ and ${\Gamma}(N,\lambda)$ denote the circularly symmetric complex Gaussian distribution with mean $\mu$ and variance $\nu$, exponential distribution with parameter $\lambda$, and gamma distribution with parameters $N$ and $\lambda$, respectively.
$\mathbb{C}^{m\times n}$ denotes the $m\times n$ complex number domain.
$\log(\cdot)$ and $\ln(\cdot)$ denote the base-2 and natural logarithms, respectively.
$f^{(m)}$ denotes the $m$-order derivative of $f$.
$[x]^{+}\triangleq \max(x,0)$.
The key symbols used in the paper are listed in Table \ref{Symbol}.
\begin{table}
	\caption{Key~Symbols~Used~in~the~Paper}
	\begin{center}\label{Symbol}
		\begin{tabular}{c|c}
			\hline
			Symbols &  Definition/Explanation \\
			\hline\hline
			$\Phi_h$, $\Phi_f$
			&  \noindent
			PPPs for Rxs in HD and FD tiers \\
			\hline
			$\hat \Phi_h$, $\hat \Phi_f$
			&  \noindent
			PPPs for Txs in HD and FD tiers \\
			\hline
			$\Phi_e$,
			&  \noindent
			PPP for eavesdroppers \\
			\hline
			$\lambda_h$, $\lambda_f$, $\lambda_e$
			&  Densities of PPPs $\Phi_h$ ($\hat \Phi_h$), $\Phi_f$ ($\hat \Phi_f$) and  $\Phi_e$\\
			\hline
			$N_h$, $N_f$
			& Numbers of antennas at HD and FD Rxs \\
			\hline
			$N_e$
			& Number of antennas at an eavesdropper\\
			\hline
			$N_t$
			& Number of jamming antennas at an FD Rx\\
			\hline
			$N_j$
			& Number of jamming signal streams at an FD Rx\\
			\hline
			$P_h$, $P_f$
			& Transmit powers at Txs in the HD and FD tiers\\
			\hline
			$P_t$
			& Power of jamming signals at an FD Rx\\
			\hline
			$(x,\hat x)$
			& Locations of a Rx and its paired Tx  \\
			\hline
			$D_h$, $D_f$
			& Distances between Tx-Rx pairs in HD and FD tiers
			\\\hline
			$D_{xy}$
			& Distance between a node at $x$ and a node at $y$\\
			\hline
			$\mathcal{P}_c$, $\mathcal{P}_t$
			& Connection probabilities of HD and FD Rxs \\
			\hline
			$\mathcal{P}_{so}$    & Secrecy outage probability of an FD Rx\\
			\hline
			$\mathcal{T}_s$
			& Network-wide secrecy throughput of the FD tier\\
			\hline
			$\mathcal{T}_c$
			& Network-wide throughput of the HD tier \\
			\hline
			\hline
		\end{tabular}
	\end{center}
\end{table}

\section{System Model}
\begin{figure}[!t]
	\centering
	\includegraphics[width=2.8in]{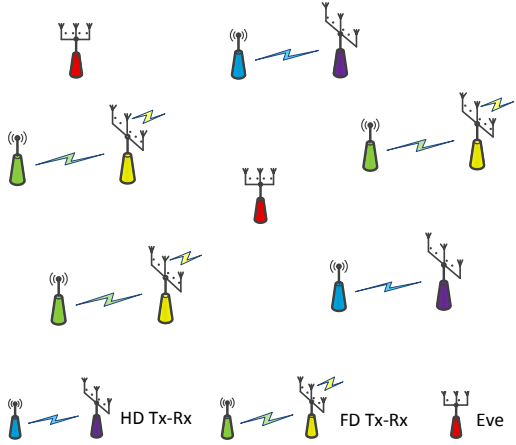}
	\caption{An illustration of a two-tier heterogeneous DWN consisting of both HD and FD tiers.
		Each HD (FD) Rx receives  data from an intended Tx.
		The ongoing transmission between the FD Tx-Rx pair is overheard by randomly located eavesdroppers (Eves).}
	\label{System_model}
\end{figure}

Consider a two-tier heterogeneous DWN in which an existing tier that provides unclassified services is overlaid with a temporarily deployed tier that has classified services.
In either tier, each data source (Tx) has only a single antenna due to hardware cost, and reports data up to its paired data collection station (Rx); each Rx is equipped with multiple antennas for signal enhancement, interference suppression, information protection, etc.
In the underlying tier, the Tx sends an unclassified message to the Rx, and the latter works in the HD mode, using all its antennas to receive the desired signal.
In the overlaid tier, the Tx deliveries a confidential message to its Rx in the presence of randomly located eavesdroppers, and the Rx works in the FD mode, simultaneously using part of its antennas to receive the desired signal and using the remaining to radiate jamming signals to confuse eavesdroppers.
An illustration of a network snapshot is depicted in Fig. \ref{System_model}.
We model the locations of HD Rxs, FD Rxs and eavesdroppers according to independent homogeneous PPPs ${\Phi}_h$ with density $\lambda_h$, ${\Phi}_f$ with density $\lambda_f$, and $\Phi_e$ with density $\lambda_e$, respectively.
We further use $\hat \Phi_h$ and $\hat \Phi_f$ to denote the sets of locations of the Txs in the HD and FD tiers, which also obey independent PPPs with densities $\lambda_h$ and $\lambda_f$ according to the displacement theorem \cite[page 35]{Haenggi2012Stochastic}.
Wireless channels are assumed to experience a large-scale path loss governed by the exponent $\alpha>2$ along with flat Rayleigh fading with fading coefficients independent and identically distributed (i.i.d.) obeying $\mathcal{CN}(0,1)$.
We assume that each Rx knows the channel state information of its paired Tx.
Since each eavesdropper passively receives signals, its channel state information is unknown, whereas its channel statistics information is available, see e.g., \cite{Zheng2014Transmission}-\cite{Zheng2013Improving}.

Without lose of generality, we consider a typical Tx-Rx pair in the FD tier and place the Rx at the origin $o$ of the coordinate system\footnote{From Slivnyak's theorem \cite{Stoyan1996Stochastic}, the spatial distribution of all the other nodes will not be affected.}.
Note that the aggregate interferences received at the typical FD Rx consist of the undesired signals from all the Txs (except for the typical Tx) and the jamming signals from itself and from all the other FD Rxs.
The received signal of the typical FD Rx is given by
\begin{align}\label{signal_fd}
\bm y_{f}=&\underbrace{\frac{\sqrt{P_f}\bm{f}_{ \hat o o}s_{f,\hat o}}{D_f^{\alpha/2}}}
_{\mathrm{desired~signal}}
+\underbrace{\sqrt{P_t}\bm{F}_{ o o}\bm{v}_{o}}
_{\mathrm{SI}} +\underbrace{\sum_{ \hat z\in\hat\Phi_h}\frac{\sqrt{P_h}\bm{f}_{\hat z o}s_{h,\hat z}}{D_{\hat z o}^{\alpha/2}}}
 _{\mathrm{HD-tier~undesired~signals}}
 \nonumber \\
&+\underbrace{\sum_{\hat z\in\hat \Phi_f\setminus \hat o}\left(\frac{\sqrt{P_f}\bm{f}_{\hat z o}s_{f,\hat z}}{D_{\hat zo}^{\alpha/2}}
	+\frac{\sqrt{P_t}\bm{F}_{z o}\bm v_{ z}}{D_{ z o}^{\alpha/2}}\right)}
_{\mathrm{FD-tier~undesired~and~jamming~signals}}
+\bm n_{f},
\end{align}
where $s_{f,\hat z}$ ($s_{h,\hat z}$) denotes the signal from the Tx located at $\hat z$ in the FD (HD) tier with $\mathbb{E}[|s_{f, \hat z}|^2]=1$ ($\mathbb{E}[|s_{h,\hat z}|^2]=1$);
$\bm v_{z}\in\mathbb{C}^{N_t\times1}$ denotes a jamming signal vector from the FD Rx at $z$ with $\mathbb{E}[\|\bm v_{z}\|^2]=1$;
$\bm n$ denotes thermal noise;
$P_f$, $P_h$ and $P_t$ denote the transmit powers of the Txs in the FD tier, in the HD tier, and of the FD Rxs, respectively;
$\bm f_{x y}\in\mathbb{C}^{(N_f-N_t)\times1}$ ($\bm F_{x y}\in\mathbb{C}^{(N_f-N_t)\times N_t}$) denotes the small-scale fading coefficient vector (matrix) of the channel from the node at $x$ to the FD Rx at $y$ ($\bm F_{o o}$ denotes the SI channel matrix related to the residual SI after passive SI suppression like antenna isolation, the entries of which can be regarded as independent Rayleigh distributed variables \cite{Krikidis2012Full}).
It is worth noting that, due to the fixed Tx-Rx pair separation distance $D_f$ \footnote{Fixing the Tx-Rx pair distance is quite generic when dealing with a DWN with or without security considerations \cite{Zhou2011Throughput,Zhang2013Enhancing,Tong2015Throughput,Hunter08Transmission}, which allows us to ease the mathematical analysis.
	Nevertheless, the results obtained under this hypothesis can be easily extended to an arbitrary distribution of $D_f$, as referred to \cite{Haenggi2009Stochastic}.}, $D_{\hat z o}$ and $D_{z o}$ in \eqref{signal_fd} are not independent; $D_{\hat z o}$ can be expressed by $D_{\hat z o}=\sqrt{D_{z o}^2+D_f^2-2D_{ z o}D_f\cos\theta_z}$, where the angle $\theta_z$ is uniformly distributed in the range $[0, 2\pi]$.
As can be seen in subsequent analysis, the correlation between $D_{\hat z o}$ and $D_{z o}$ makes it challenging to derive analytically tractable results for involved performance metrics.

As to the HD Rx located at $b$, since it suffers no SI, the received signal is given by
\begin{align}\label{signal_hd}
\bm y_{h}
=&{\frac{\sqrt{P_h}\bm{h}_{\hat o o}s_{h,\hat o}}{D_h^{\alpha/2}}}
+{\sum_{\hat z\in\hat \Phi_h\setminus \hat o}\frac{\sqrt{P_h}\bm{h}
_{ \hat z o}s_{h,\hat z}}{D_{\hat z o}^{\alpha/2}}}\nonumber \\
&
+{\sum_{ \hat z\in\hat \Phi_f}\left(\frac{\sqrt{P_f}\bm{h}_{ \hat z o}s_{f,\hat z}}{D_{\hat z o}^{\alpha/2}}
+\frac{\sqrt{P_t}\bm{H}_{z o}\bm v_{ z}}{D_{z o}^{\alpha/2}}\right)}
+\bm n_{h},
\end{align}
where $\bm h_{x y}\in\mathbb{C}^{N_h\times1}$ ($\bm H_{x y}\in\mathbb{C}^{N_h\times N_t}$) denotes the small-scale fading coefficient vector (matrix) of the channel from the node at $x$ to the HD Rx at $y$.

Similarly, for the eavesdropper located at $e$ that is intended to wiretap the data transmission from the
typical Tx to the typical FD Rx, the received signal is given by
\begin{align}\label{signal_eve}
\bm y_{e}
=&
{\frac{\sqrt{P_f}\bm{g}_{\hat o e}s_{f, \hat o}}{D_{\hat oe}^{\alpha/2}}}
+\frac{\sqrt{P_t}\bm{G}_{ oe}\bm{v}_{ o}}
{D_{oe}^{\alpha/2}}
+{\sum_{\hat z\in\hat \Phi_h}\frac{\sqrt{P_h}\bm{g}_{ \hat ze}s_{h,\hat z}}{D_{\hat ze}^{\alpha/2}}}
\nonumber\\
&+{\sum_{\hat z\in\hat \Phi_f\setminus \hat o}\left(\frac{\sqrt{P_f}\bm{g}_{ \hat ze}s_{f,\hat z}}{D_{\hat ze}^{\alpha/2}}
	+\frac{\sqrt{P_t}\bm{G}_{ze}\bm v_{ z}}{D_{ze}^{\alpha/2}}\right)}
+\bm n_{e},
\end{align}
with $\bm g_{x e}\in\mathbb{C}^{N_e\times1}$ ($\bm G_{x e}\in\mathbb{C}^{N_e\times N_t}$) the fading coefficient vector (matrix) of the link from the node at $x$ to the eavesdropper at $e$ and ${\sqrt{P_t}\bm{G}_{ oe}\bm{v}_{ o}}
{D_{oe}^{-\alpha/2}}$ the interference from the typical FD Rx.

\subsection{Wiretap Encoding and Performance Metrics}
We consider a non-colluding wiretap scenario in which eavesdroppers do not cooperate with each other and each eavesdropper individually decodes a secret message.
To safeguard information security, we utilize the well-known Wyner's wiretap encoding scheme \cite{Wyner1975Wire-tap} to encode secret information.
Let ${R}_{t}$ and ${R}_{s}$ denote the rates of the transmitted codewords and the embedded secret messages, and
${R}_e\triangleq {R}_t -{R}_s$ denote the rate of redundant information that is exploited to provide secrecy against eavesdropping.
If a Tx-Rx link can support the rate ${R}_t$, the Rx is able to decode the secret messages; this corresponds to a reliable connection event \cite{Zhou2011Throughput}.
The connection probability of a typical FD Rx is defined as the probability that the signal-to-interference-plus-noise ratio (SINR) of the FD Rx, denoted by $\textsf{SINR}_f$, lies above an SINR threshold $\beta_t\triangleq2^{R_t}-1$, i.e.,
\begin{equation}\label{pt_fd_def}
\mathcal{P}_t\triangleq \mathbb{P}\{\textsf{SINR}_f > \beta_t\}.
\end{equation}
Similarly, the connection probability of an HD Rx is defined by $\mathcal{P}_c\triangleq\mathbb{P}
\{\textsf{SINR}_h>\beta_c\}$, where $\beta_c\triangleq 2^{{R}_c}-1$ with ${R}_c$ the corresponding codeword rate.

If the channel from the Tx to any of the eavesdroppers can support the redundant rate ${R}_e$, perfect secrecy is compromised and a secrecy outage event occurs \cite{Zhang2013Enhancing}.
The secrecy outage probability is defined as the complement of the probability that the SINR of an arbitrary eavesdropper at $e$, denoted by $\textsf{SINR}_{e}$, lies below an SINR threshold $\beta_e\triangleq2^{R_e}-1$, i.e.,
\begin{equation}\label{ps_fd_def}
\mathcal{P}_{so}\triangleq 1-\mathbb{E}_{\Phi_e}
\left[\prod_{e\in\Phi_e}\mathbb{P}
\left\{\textsf{SINR}_{e}
<\beta_{e}|\Phi_e\right\}\right].
\end{equation}

To evaluate the efficiency of secure transmissions in a DWN, we focus on the performance metric named \emph{network-wide secrecy throughput} (unit: $\mathrm{bits/s/Hz/m^2}$), which is defined as the averagely successfully transmitted information bits per second per Hertz per unit area under a connection probability $\mathcal{P}_{t}(\beta_{t})=\sigma$ and a secrecy outage probability $\mathcal{P}_{so}(\beta_{e})=\epsilon$ \cite{Zhou2011Throughput}, \cite{Zhang2013Enhancing}, i.e.,
\begin{align}\label{st_def}
\mathcal{T}_{s}
\triangleq \lambda_f \sigma\mathcal{R}^*_{s}
&=\lambda_f \sigma
\left[\mathcal{R}^*_{t} - \mathcal{R}^*_{e}\right]^+
\nonumber\\
&= \lambda_f \sigma
\left[\log\left({1+\beta^*_{t}}\right)
-\log\left(
{1+\beta^*_{e}}\right)\right]^+.
\end{align}
In \eqref{st_def},  ${R}^*_{t}\triangleq \log (1+\beta^*_{t})$, ${R}^*_{e}\triangleq \log (1+\beta^*_{e})$ and ${R}^*_{s}={R}^*_{t}-{R}^*_{e}$ denote the codeword rate, redundant rate and secrecy rate at a Tx in the FD tier, with $\beta^*_{t}$ and $\beta^*_{e}$ satisfying $\mathcal{P}_{t}(\beta_{t})=\sigma$
and $\mathcal{P}_{so}(\beta_{e})=\epsilon$, respectively.
Likewise, the network-wide throughput of the HD tier under a connection probability $\mathcal{P}_c(\beta_c)=\sigma_c$ is defined by
$\mathcal{T}_{c}\triangleq \lambda_h\sigma_c R_c^*$, where $R_c^*\triangleq \log (1+\beta_c^*)$ with $\beta_c^*$ satisfying $\mathcal{P}_{c}(\beta_c)=\sigma_c$.

We emphasize that the FD tier density strikes a non-trivial tradeoff between spatial reuse, reliable connection, and safeguarding.
On one hand, increasing the density of the FD tier establishes more communication links per unit area, potentially increasing throughput;
meanwhile, the increased jamming signals introduced by newly deployed FD Rxs greatly degrade the wiretap channels.
On the other hand, the additional amount of interference caused by adding new devices deteriorates ongoing receptions, decreasing the probability of successfully connecting Tx-Rx pairs.
The overall balance of such opposite effects on secrecy throughput needs to be carefully addressed.
Note that, from a network perspective, network designers may concern themselves more with the network deployment rather than optimizing other parameters like transmit or jamming power, antenna number, etc.
For example, in a security monitoring wireless sensor network, network designers may fix the transmit power of sensor nodes to their maximum values for simplicity, but would modestly design how many sensors should be scattered in order to satisfy the monitoring requirement.
Therefore, in this paper, we aim to determine the deployment of the FD tier to achieve the maximum network-wide secrecy throughput while guaranteeing a certain level of network-wide throughput for the HD tier.

In the following sections, we deal with network design by considering the scenarios of each FD Rx using single-antenna (SA) jamming ($N_t=1$) and using multi-antenna (MA) jamming ($N_t>1$), respectively.
The reason behind such a division is threefold:

1)
In the SA case, SSIC can only be operated at the FD Rx's input using a hybrid zero-forcing and maximal ratio combining (ZF-MRC) criterion;
in the MA case, due to the extra degrees of freedom, the FD Rx simply performs SSIC at the output  and just adopts the MRC reception at the input.

2)
In the SA case, since the channel from the FD receiver's output to either the legitimate node or to the eavesdropper is a  single-input multi-output channel, it is relatively easy to analyze connection probability and secrecy outage probability;
in the MA case, either of the above channels is an MIMO channel, which makes the analysis much more complicated, e.g., analyzing the secrecy outage probability requires using the theory from integer partitions.

3)
In the SA case, we prove the quasi-concavity of network-wide secrecy throughput w.r.t. the FD tier density, and calculate the optimal FD tier density that maximizes network-wide secrecy throughput using the bisection method;
in the MA case, due to the analytically intractable integer partitions,
we can only obtain the peak network-wide secrecy throughput via one-dimensional exhaustive search.

In our analysis, due to uncoordinated concurrent transmissions, the  aggregate interference at a receiver dominates thermal noise.
For tractability, we consider the \emph{interference-limited} case by ignoring thermal noise.
Nevertheless, our results can be easily generalized to the case with thermal noise.
For ease of notation, we define $\delta\triangleq {2}/{\alpha}$,
$C_{\alpha,N}\triangleq \frac{\pi\Gamma\left(N-1+\delta\right)
	\Gamma\left(1-\delta\right)}
{\Gamma\left(N-1\right)}$, and $P_{ab}\triangleq {P_a}/{P_b}$ for $a,b\in\{h,f,t\}$.

\section{Single-antenna-jamming FD Receiver}
In this section, we consider the scenario where each FD Rx uses single-antenna jamming, i.e., $N_t=1$.
Thereby, matrices $\bm F$, $\bm H$ and $\bm G$ given in \eqref{signal_fd}-\eqref{signal_eve} reduce to vectors ${\bm f}$, $ {\bm h}$, and ${\bm g}$, respectively, and vector $\bm v$ reduces to scalar $v$.
Without lose of generality, we consider a typical FD Tx-Rx pair $(\hat o,o)$.
We first analyze the connection probability and the secrecy outage probability of the typical FD Rx, and then maximize network-wide secrecy throughput by optimizing the density of the FD tier.

To counteract the SI and simultaneously strengthen the desired signal, the weight vector $\bm w_{f}$ at the FD Rx's input can be chosen according to a hybrid ZF-MRC criterion\footnote{ZF-MRC receiver might not be the optimal one, but it is really a simple yet efficient alternative that yields an achievable secrecy rate/throughput.
	Limiting the receiver in the null space of SI successfully avoids the SI invasion in the spatial domain without consuming extra circuit power that would be used for SI cancellation in the analog or digital domain.}, which is
\begin{equation}\label{zf_mrc}
\bm w_{f}= \frac{\bm f_{ \hat o o}^H \bm U\bm U^H}{\|\bm f_{ \hat o o}^H \bm U\|},
\end{equation}
where $\bm U\in\mathbb{C}^{(N_f-1)\times(N_f-2)}$ is the projection matrix onto the null space of vector $\bm f_{o o}^H$ such that the columns of $\left[\frac{\bm f_{oo}^H}{\|\bm f_{o o}\|}, \bm U\right]$ constitute an orthogonal basis.
In this way, we have $\bm w_{f}{\bm f}_{ oo}=0$.

\subsection{Connection Probability}
In this subsection, we investigate the connection probability of the typical FD Rx.
From \eqref{signal_fd} and \eqref{zf_mrc}, the SIR of the typical FD Rx is given by
\begin{align}\label{sir_fd}
\textsf{SIR}_{f} = \frac{P_f\|\bm f_{ \hat o o}^H \bm U\|^2D_f^{-\alpha}}
{I_h+I_f},
\end{align}
where $I_h\triangleq \sum_{\hat z\in\hat \Phi_h}P_h |\bm w_{f} \bm f_{\hat z o}|^2 D_{\hat z o}^{-\alpha}$ and $I_f\triangleq\sum_{\hat z\in\hat \Phi_f\setminus \hat o}\left(P_f |\bm w_{f} \bm f_{\hat z o}|^2 D_{ \hat z o}^{-\alpha}+P_t|\bm w_{f} \bm f_{z o}|^2D_{ zo}^{-\alpha}\right)$ are the aggregate interferences from the HD tier and FD tier, respectively.
In the following theorem, we provide a general expression of the exact connection probability $\mathcal{P}_t$.
\begin{theorem}\label{pt_exact_theorem}
The connection probability of a typical FD Rx is
\begin{align}\label{pt_exact}
\mathcal{P}_{t}
= \sum_{m=0}^{N_f-3}&
\frac{(-1)^m}{m!}
\left(\frac{D_f^{\alpha}\beta_t}{P_f}\right)^m\times\nonumber\\
&
\left(e^{-\lambda_hC_{\alpha,2}(P_hs)^{\delta}}\mathcal{L}_{I_f}\left({D_f^{\alpha}\beta_t}/{P_f}\right)\right)^{(m)} ,
\end{align}
where
$\mathcal{L}_{I_f}(s)$ denotes the Laplace transform of $I_f$, i.e.,
\begin{align}\label{laplace}
    \mathcal{L}_{I_f}(s)
    &=\exp\Bigg(-\lambda_f\int_0^{\infty}
    \Bigg(2\pi-\int_{o}^{2\pi}\frac{1}{1+P_f sr^{-\alpha}}\times \nonumber\\
  &  \frac{d\theta}
    {1+P_ts(r^2+D_f^2-2rD_f \cos\theta)^{-\alpha/2}}\Bigg)rdr\Bigg).
\end{align}
\end{theorem}
\begin{proof}
Please see Appendix \ref{appendix_pt_exact_theorem}.
\end{proof}

Theorem \ref{pt_exact_theorem} provides an exact connection probability without
requiring time-consuming Monte Carlo simulations.
A special case is that when $N_f = 3$, $\mathcal{P}_{t}$ simplifies to $e^{-\lambda_hC_{\alpha,2}
	(P_hs)^{\delta}}\mathcal{L}_{I_f}
    \left({D_f^{\alpha}\beta_t}/{P_f}\right)$.
    However, for the more general case, the double integral term in \eqref{laplace} makes computing $\mathcal{L}^{(m)}_{I_f}(s)$ quite difficult, thus making \eqref{pt_exact} rather unwieldy to analyze.
    This motivates the need for more compact forms, and in the following theorem we provide closed-form lower and upper bounds for $\mathcal{P}_t$.
\begin{theorem}\label{pt_bound_theorem}
    The connection probability $\mathcal{P}_{t}$ of a typical FD Rx is lower bounded by $\mathcal{P}_{t}^L$ and upper bounded by $\mathcal{P}_{t}^U$;
    which share the same closed form given below,
\begin{align}\label{pt_bound}
\mathcal{P}_{t}^S = e^{-\Lambda_{f}^S\beta_t^{\delta}}
&+e^{-\Lambda_{f}^S\beta_t^{\delta}}\sum_{m=1}^{N_f-3}\frac{1}{m!}\times\nonumber\\
&\sum_{n=1}^{m}\left(
\delta\Lambda_{f}^S\beta_t^{\delta}\right)^n
\Upsilon_{m,n}, ~\forall S\in\{L,U\}
\end{align}
where $\Lambda_{f}^L\triangleq C_{\alpha,2}D_f^2\left({P_{hf}^{\delta}}
\lambda_h+\left(1+{P_{tf}^{\delta}}
\right)\lambda_f\right)$, $\Lambda_{f}^U\triangleq C_{\alpha,2}D_f^2\left({P_{hf}^{\delta}}\lambda_h+\frac{ 1+\delta}{2}\left(1+{P_{tf}^{\delta}}
\right)\lambda_f\right)$ and
$\Upsilon_{m,n}=\sum_{\psi_j\in \mathrm{comb}\binom{m-1}{m-n}}\prod_
{
	\substack{
		l_{ij}\in\psi_j\\
		i=1,\cdots,m-n}
}
\left(l_{ij}-\delta(l_{ij}-i+1)\right)$.
Here $\mathrm{comb}\binom{m-1}{m-n}$ denotes the set of all distinct subsets of the natural numbers $\{1,2,\cdots,m-1\}$ with cardinality $m-n$.
The elements in each subset are arranged in an increasing order with $l_{ij}$ the $i$-th element of $\psi_j$.
For $m\ge1$, we have $\Upsilon_{m,m}=1$.
\end{theorem}
\begin{proof}
Please see Appendix \ref{appendix_pt_bound_theorem}.
\end{proof}

Although \eqref{pt_bound} still seems  complicated, it is actually very easy to compute without any integrals.
Considering a practical need of a high level of reliability, we concentrate on the
large probability region in which $\mathcal{P}_{t}^S\rightarrow 1$ for $S\in\{L,U\}$, and provide a much simpler approximation for $\mathcal{P}_{t}^S$ in the following corollary, which further facilitates the analysis.
\begin{corollary}\label{pt_high_corollary}
	In the large probability region, i.e., $\mathcal{P}_t\rightarrow 1$, $\mathcal{P}_t^S$ in \eqref{pt_bound} is approximated by
	\begin{equation}\label{pt_high}
	\mathcal{P}_t^S\approx 1-\Lambda_{f}^S\beta_t^{\delta}
	K_{\alpha,N_f-2},  ~\forall S\in\{L,U\}
	\end{equation}
	where  $K_{\alpha,N}=1+\sum_{m=1}^{N-1}\frac{1}{m!}
	\prod_{l=0}^{m-1}(l-\delta)$.
\end{corollary}
\begin{proof}\label{pt_high_proof}
	We see from \eqref{pt_bound} that $\mathcal{P}_{t}^S\rightarrow 1$ as $\Lambda_{f}^S\rightarrow 0$.
	Here, $\Lambda_{f}^S\rightarrow 0$ reflects all cases of system parameters such as $D_f$, $\lambda_f$ and $\lambda_h$ that may lead to a large $\mathcal{P}_{t}^S$.
	A reasonable case of $\Lambda_{f}^S\rightarrow 0$ is but is not limited to that the Tx-Rx pair distance is much less than the average distance between any two Txs (or between two Rxs), i.e., $D_f^2\lambda_f, D_f^2\lambda_h\ll 1$.
	Using the first-order Taylor expansion with  \eqref{pt_bound} around $\Lambda_{f}^S=0$ and discarding the high order terms $\Theta \left(\left(\Lambda_{f}^S\right)^2\right)$, we complete the proof.
\end{proof}

\begin{figure}[!t]
	\centering
	\includegraphics[width=2.8in]{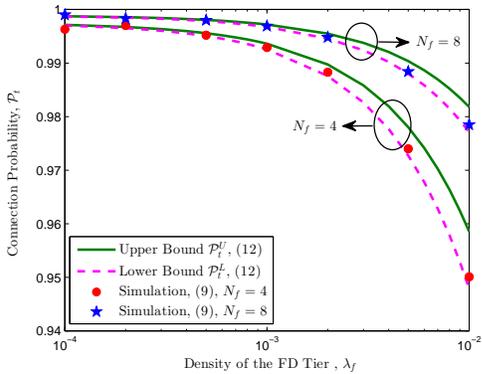}
	\caption{Connection probability vs. $\lambda_f$ for different values of $N_f$, with $P_t = 0$dBm, $N_t=1$, and $\beta_t = 1$.
		Unless specified otherwise, we set $\alpha=3.5$, $P_f=P_h= 0$dBm, $N_h=4$, $\lambda_h=10^{-3}$ and $D_f=D_h=1$.}
	\label{PT_LF_NF}
\end{figure}

The bound results given in Corollary \ref{pt_high_corollary} are shown in Fig. \ref{PT_LF_NF}, in which we see in the large probability region the lower bound $\mathcal{P}_{t}^L$ is tight to the exact value of $\mathcal{P}_{t}$ from Monte-Carlo simulations.
Therefore in subsequent analysis, we focus on the lower bound $\mathcal{P}_{t}^L$ instead of $\mathcal{P}_{t}$.
Note that this is actually a pessimistic evaluation of real connection performance.
From Fig. \ref{PT_LF_NF}, we also find that connection performance deteriorates as the FD tier density $\lambda_f$ increases due to the additional amount of interference.
This is ameliorated by adding receive antennas at the FD Rx.

Comparing \eqref{signal_hd} with \eqref{signal_fd}, it is not difficult to conclude that the connection probability $\mathcal{P}_{c}$ of a typical HD Rx shares a similar form as $\mathcal{P}_{t}$.
Likewise, we can obtain an approximation for $\mathcal{P}_{c}$ in the large probability region, which is provided by the following corollary.
\begin{corollary}\label{pc high_corollary}
	Define $\Lambda_{h}\triangleq C_{\alpha,2}\left(\lambda_h+
	\left(P_{fh}^{\delta}+P_{th}^{\delta}\right)
	\lambda_f\right)D_h^2$.
	In the large probability region, the connection probability $\mathcal{P}_c$ of a typical HD Rx is approximated by
	\begin{equation}\label{pc_high}
	{\mathcal{P}}_{c} \approx 1-\Lambda_{h}\beta_c^{\delta}
	K_{\alpha,N_h}.
	\end{equation}
\end{corollary}

\subsection{Secrecy Outage Probability}
In this subsection, we investigate the secrecy outage probability which corresponds to the probability that a secret message is decoded by \emph{at least} one eavesdropper in the network.

To guarantee secrecy, we should not underestimate the wiretap capability of eavesdroppers.
Thereby, we consider a worst-case wiretap scenario assuming that all eavesdroppers have multiuser decoding capabilities (e.g., successive interference cancellation), thus each eavesdropper itself can successfully resolve the signals radiated from those unexpected transmitters and remove them from its received signals\footnote{Successfully decoding multiplex signals actually depends on the so-called `capacity range', which is extremely hard to analyze if we have more than two users. In this paper, we are not going to spend time to analyze the complicated multiplex channel capacity, instead we consider the eavesdropping capacity of successive interference cancellation, which is actually the worse case scenario.} \cite{Zhang2013Enhancing}.
In this way, the aggregate interference received at each eavesdropper only consists of the jamming signals emitted by all the FD Rxs.
We assume that eavesdroppers use the optimal linear receiver, i.e., the minimum mean square error (MMSE) receiver \cite{Gao1998Theoretical}, to strengthen the received signals.
From \eqref{signal_eve}, the weight vector of the eavesdropper located at $e$ when using the MMSE receiver is
\begin{equation}\label{weight_mmse}
\bm w_{e} = \bm R_{e} ^{-1}\bm g_{\hat oe},
\end{equation}
where $\bm R_{e} \triangleq P_t\bm g_{oe}\bm g_{ oe}^HD_{o e}^{-\alpha}+\sum_{ z\in\Phi_f\setminus o}P_t\bm g_{ze}\bm g_{ze}^HD_{z e}^{-\alpha}$, and the resulting SIR is given by
\begin{equation}\label{sir_e}
\textsf{SIR}_{e} = P_f\bm g_{\hat oe}^H\bm R_{e}^{-1}\bm g_{\hat oe}D_{\hat o e}^{-\alpha}.
\end{equation}
In the following theorem, we provide a general expression of the exact secrecy outage probability.
\begin{theorem}\label{ps_exact_theorem}
    The secrecy outage probability $\mathcal{P}_{so}$ of a typical FD Rx is
\begin{align}\label{ps_exact}
    \mathcal{P}_{so}
    = 1 - &\exp\Bigg(
    -\lambda_e\sum_{n=0}^{N_e-1}
    \sum_{i=0}^{\min(n,1)}
    \frac{\left(C_{\alpha,2}\lambda_f
    \left(P_{tf}\beta_e\right)^{\delta}\right)
    ^{n-i}}
    {(n-i)!}\nonumber\\
&        \int_0^{\infty}
    \mathcal{Q}_i(r)r^{2(n-i)}
    e^{-C_{\alpha,2}\lambda_f
    \left(P_{tf}\beta_e\right)
    ^{\delta}r^2}rdr\Bigg),
    \end{align}
where $\mathcal{Q}_i(r)=\int_0^{2\pi}\frac{\left(P_{tf}\beta_e\left({r}/
{\sqrt{r^2+D_f^2-2rD_f\cos\theta}}\right)^{\alpha}\right)^i}
{1+P_{tf}\beta_e\left({r}/
{\sqrt{r^2+D_f^2-2rD_f\cos\theta}}\right)^{\alpha}}d\theta$.
\end{theorem}
\begin{proof}\label{ps_exact_proof}
    Please see Appendix \ref{appendix_ps_exact_theorem}.
\end{proof}

The double integral in \eqref{ps_exact} makes $\mathcal{P}_{so}$ difficult to analyze.
Note that, in a large-scale DWN, a Tx is generally a simple node with low transmit power, e.g., a sensor, and has very short coverage.
Therefore, to guarantee a reliable communication and meanwhile protect information from eavesdropping, the separation distance $D_f$ between a Tx-Rx pair is usually set small.
In view of this, in order to develop useful and meanwhile tractable insights into the behavior of the secrecy outage probability $\mathcal{P}_{so}$, we resort to an asymptotic analysis by considering a small $D_f$ regime, e.g., $D_f\rightarrow 0$, and provide quite a simple approximation for $\mathcal{P}_{so}$ in Corollary \ref{ps_approx_corollary}.
We stress that, although Corollary \ref{ps_approx_corollary} is obtained by assuming $D_f\rightarrow 0$, it applies more generally.
As can be seen in Fig. \ref{PSO_APPROXIMATION}, \eqref{ps_approx} provides a very accurate approximation for the exact value in \eqref{ps_exact} for quite a wide range of $D_f$.
This illustrates the rationality of the given hypothesis.
Hereafter, unless specified otherwise, we focus on the case $D_f\rightarrow 0$ when referring to the secrecy outage probability $\mathcal{P}_{so}$ of the typical FD Rx.

\begin{corollary}\label{ps_approx_corollary}
    In the small $D_f$ regime, i.e., $D_f\rightarrow 0$,   $\mathcal{P}_{so}$ in \eqref{ps_exact} is approximated by
    \begin{equation}\label{ps_approx}
        \mathcal{P}_{so}\approx 1-\exp\left[-\frac{\pi\lambda_e}
        {C_{\alpha,2}\lambda_fP_{tf}^{\delta}
        \beta_e^{\delta}}\left(
        N_e-1+\frac{1}{1+P_{tf}\beta_e}\right)\right].
    \end{equation}
\end{corollary}
\begin{proof}
    Please see Appendix \ref{appendix_ps_approx_corollary}.
\end{proof}

\begin{figure}[!t]
\centering
\includegraphics[width=2.5in]{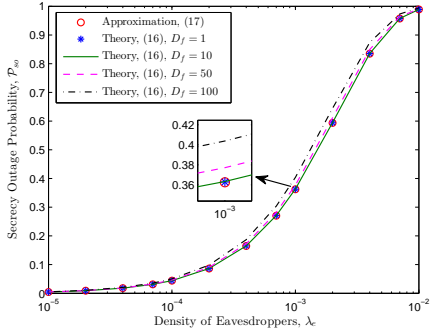}
\caption{Secrecy outage probability vs. $\lambda_e$ for different values of $D_f$, with $P_t = 10$dBm, $N_t=1$, $N_e=4$, $\lambda_f=10^{-3}$ and $\beta_e = 1$.}
\label{PSO_APPROXIMATION}
\end{figure}

\begin{figure}[!t]
\centering
\includegraphics[width=2.5in]{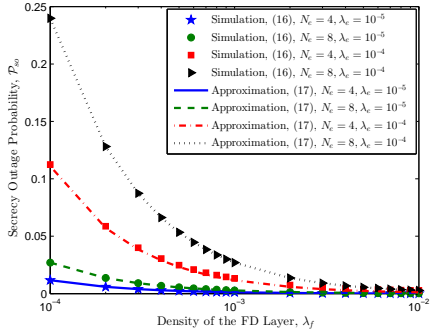}
\caption{Secrecy outage probability vs. $\lambda_f$ for different values of $N_e$ and $\lambda_e$, with $P_t = 20$dBm, $N_t=1$ and $\beta_e = 1$.}
\label{PSO_LF_NE_LE}
\end{figure}

Fig. \ref{PSO_LF_NE_LE} shows that the approximated values in \eqref{ps_approx} for the secrecy outage probability $\mathcal{P}_{so}$ are quite close to the exact results of Monte-Carlo simulations.
The value of $\mathcal{P}_{so}$ gets larger as either the number $N_e$ of eavesdropper's antennas or the density $\lambda_e$ of eavesdroppers  increases.
To reduce the secrecy outage probability, we should better deploy more FD jammers, i.e., increasing the value of $\lambda_f$.

Although eavesdroppers do not collude with each other, they may use large antennas for better wiretapping.
Considering that the value of $N_e$ goes to infinity, $\mathcal{P}_{so} $ in \eqref{ps_approx} reduces to
\begin{equation}\label{ps_large_antenna}
\mathcal{P}_{so} = 1-\exp\left(-\frac{\pi\lambda_eN_e}
{C_{\alpha,2}\lambda_fP_{tf}^{\delta}
	\beta_e^{\delta}}\right).
\end{equation}
We observe from \eqref{ps_large_antenna} that  $\mathcal{P}_{so}$ increases as $\alpha$ increases.
This is because, in an environment of a larger path loss, jamming signals have undergone stronger attenuation before they arrive at the eavesdroppers.

\subsection{Network-wide Secrecy Throughput}
In this subsection, we investigate network-wide secrecy throughput  $\mathcal{T}_{s}$ under a connection probability $\mathcal{P}_{t}(\beta_{t})=\sigma$ and a secrecy outage constraint $\mathcal{P}_{so}(\beta_{e})=\epsilon$, which is defined in \eqref{st_def}.
Clearly, if $\beta^*_{t} \le\beta^*_{e}$, a positive $\mathcal{T}_{s}$ that simultaneously satisfies the given probabilities does not exists, thus, transmissions should be suspended.
Note that although increasing the density $\lambda_f$ of the FD tier may increase network-wide secrecy throughput $\mathcal{T}_s$, it introduces greater interference to the HD tier, thus reducing network-wide throughput $\mathcal{T}_c$.
To achieve a good balance, we should carefully choose the value of $\lambda_f$.
In the following, we aim to maximize
$\mathcal{T}_s$ by optimizing $\lambda_f$ under a guarantee that $\mathcal{T}_c$ lies above a target throughput $T_c$.
This optimization problem is formulated as
\begin{align}\label{st_max_def}
\max_{\lambda_f}~ \mathcal{T}_{s},\quad
\mathrm{s.t.}~ \mathcal{T}_{c}\ge {T}_c.
\end{align}
To proceed, we first calculate $\beta^*_{t}$ and $\beta^*_{e}$ from $\mathcal{P}_{t}(\beta_{t})=\sigma$
and $\mathcal{P}_{so}(\beta_{e})=\epsilon$, respectively.
In general, the analytical expressions of the exact $\beta^*_{t}$ and $\beta^*_{e}$ are unavailable due to the complexity of \eqref{pt_exact} and \eqref{ps_exact}; we can only numerically calculate $\beta^*_{t}$ and $\beta^*_{e}$, which makes solving problem \eqref{st_max_def} extremely difficult.
To facilitate the analysis and provide useful insights into network design, we resort to some approximate results of the connection and secrecy outage probabilities.
To ensure a high level of reliability, the connection probability $\sigma$ is expected to be large, which allow us to use \eqref{pt_high} to calculate $\beta^*_{t}$.
\begin{lemma}\label{beta_t_lemma}
	In the large $\sigma$ regime, i.e., $\sigma\rightarrow 1$, the root $\beta_{t}$ of the equation $\mathcal{P}_{t}(\beta_{t})=\sigma$ is given by
	\begin{equation}\label{beta_t}
	\beta_t^* = \left(\frac{1-\sigma}
	{C_{\alpha,2}D_f^2 K_{\alpha,N_f-2}\left(
		P_{hf}^{\delta}\lambda_h+(1+P_{tf}^{\delta})
		\lambda_f\right)}\right)^{\frac{\alpha}{2}}.
	\end{equation}
\end{lemma}
\begin{proof}
	Recalling Corollary \ref{pt_high_corollary}, we obtain \eqref{beta_t} by solving the equation $1-\Lambda_{f}^L\beta_t^{\delta}
	K_{\alpha,N_f-2}=\sigma$.
\end{proof}


Considering the scenario of large-antenna eavesdroppers, i.e., $N_e\gg1$, we have the following lemma.
\begin{lemma}\label{beta_t_lemma}
	In the large $N_e$ regime,  the root $\beta_{e}$ of the equation $\mathcal{P}_{so}(\beta_{e})=\epsilon$ is given by
	\begin{equation}\label{beta_e}
	\beta_e^* = \frac{1}{P_{tf}}
	\left(\frac{\pi\lambda_eN_e}
	{C_{\alpha,2}\lambda_f\ln\frac{1}{1-\epsilon}}
	\right)^{\frac{\alpha}{2}}.
	\end{equation}
\end{lemma}
\begin{proof}
	Recalling \eqref{ps_large_antenna}, we obtain \eqref{beta_e} by solving the equation
	$1-e^{-{\pi\lambda_e N_e}/\left({C_{\alpha,2}\lambda_fP_{tf}^{\delta}        \beta_e^{\delta}}\right)}=\epsilon$.
\end{proof}

Having obtained $\beta_t^*$ in \eqref{beta_t} and $\beta_e^*$ in \eqref{beta_e}, we begin to solve problem \eqref{st_max_def}.
Since we focus on variable $\lambda_f$, we substitute \eqref{st_def} into \eqref{st_max_def}, and reform problem \eqref{st_max_def} as follows by introducing an auxiliary function $F(\lambda_f)$ which shows explicitly the relationship between objective function $\mathcal{T}_s$ and $\lambda_f$,
\begin{align}\label{st_max}
\max_{\lambda_f}~ \mathcal{T}_{s}= \frac{\sigma}{\ln2}
\left[F(\lambda_f)\right]^+,\quad
\mathrm{s.t.} ~0<\lambda_f\le \lambda_f^U,
\end{align}
where $F(\lambda_f)= \lambda_f
\ln\frac{1+X(1+Y\lambda_f)^{-\alpha/2}}
{1+Z\lambda_f^{-\alpha/2}}$, $\lambda_f^U
\triangleq\frac{{({1-\sigma_c})/\left({
			C_{\alpha,2}D_h^2K_{\alpha,N_h}}\right)\left(2^{{T_c}/
			({\lambda_h\sigma_c})}-1\right)^{-\delta}-\lambda_h}}{{P_{th}^{\delta}+P_{fh}^{\delta}}}
$
is obtained from $\mathcal{T}_{c}=\lambda_h\sigma_c\log(1+\beta_c^*)=T_c$ in \eqref{st_max_def}, and  $X\triangleq\left(\frac{1-\sigma}{C_{\alpha,2}D_f^2
	K_{\alpha,N_f-2}P_{hf}^{\delta}\lambda_h}\right)
^{{\alpha}/{2}}$, $Y\triangleq\frac{1+P_{tf}^{\delta}}{P_{hf}^{\delta}\lambda_h}$ and $Z\triangleq\frac{1}{P_{tf}}\left(
\frac{\pi\lambda_eN_e}{C_{\alpha,2}
	\ln\frac{1}{1-\epsilon}}\right)
^{{\alpha}/{2}}$.
To achieve a positive $\mathcal{T}_{s}$ in \eqref{st_max}, $F(\lambda_f)>0$, i.e., $X(1+Y\lambda_f)^{-\alpha/2}>Z\lambda_f^{-\alpha/2}$ must be guaranteed, and thus we have $\lambda_f>\lambda_f^L\triangleq 1/
\left(\left({X}/{Z}\right)^{\delta}-Y\right)$ and $\left({X}/{Z}\right)>Y^{\frac{\alpha}{2}}$, which further yield
\begin{align}\label{sigma_epsilon}
(1-\sigma)\ln\frac{1}{1-\epsilon}>
\pi\lambda_eN_eD_f^2K_{\alpha,N_f-2}
\left(1+P_{tf}^{-\delta}\right).
\end{align}
That is to say, a large $\sigma$ and a small $\epsilon$ might not be simultaneously promised.
In the following, we consider the case that a positive $\mathcal{T}_{s}$ exists, i.e., $\lambda_f>\lambda_f^L$.
Thereby, problem \eqref{st_max} is equivalent to
\begin{align}\label{F_max}
\max_{\lambda_f}~
F(\lambda_f), \quad
\mathrm{s.t.}~ \lambda_f^L<\lambda_f\le \lambda_f^U.
\end{align}
In the following theorem, we can prove the \emph{quasi-concavity} \cite[Sec. 3.4.2]{Boyd2004Convex} of $F(\lambda_f)$ w.r.t. $\lambda_f$ in the range $(\lambda_f^L,\infty)$, and derive the optimal $\lambda_f$ that maximizes $F(\lambda_f)$ (or $\mathcal{T}_{s}$).
\begin{theorem}\label{opt_lambda_f_theorem}
	The optimal density $\lambda_f$ that maximizes network-wide secrecy throughput $\mathcal{T}_{s}$ is
	\begin{equation}\label{opt_lambda_f}
	\lambda_f^* =
	\begin{cases}
	~\min(\lambda^{\star},\lambda_f^U),& \left({X}/{Z}\right)>Y^{\frac{\alpha}{2}}~ \textrm{and} ~ \lambda_f^L\le\lambda_f^U,\\
	~\varnothing, & \textrm{otherwise},
	\end{cases}
	\end{equation}
	where $\lambda^{\star}$ is the unique root of the following equation,
	\begin{equation}\label{opt_lambda_exp}
	\ln\frac{f_1(\lambda)}{f_2(\lambda)}
	+\frac{\frac{\alpha}{2}f_1(\lambda)[f_2(\lambda)-1]
		-\frac{\alpha}{2}\lambda[f_1(\lambda)-f_2(\lambda)]Y}
	{f_1(\lambda)f_2(\lambda)(1+\lambda Y)}=0,
	\end{equation}
	with $f_1(\lambda)= 1+X(1+Y\lambda)^{-\frac{\alpha}{2}}$ and $f_2(\lambda)= 1+Z\lambda^{-\frac{\alpha}{2}}$.
	The left-hand side (LHS) of \eqref{opt_lambda_exp} is first positive and then negative; thus, the value of $\lambda^{\star}$ can be efficiently calculated using the bisection method.
	Here, $\lambda_f^* = \varnothing$ means no $\lambda_f$ can produce a positive $\mathcal{T}_s$ under a given pair $(\sigma,\epsilon)$.
\end{theorem}
\begin{proof}
	Please see Appendix \ref{appendix_opt_lambda_f_theorem}.
\end{proof}

Theorem \ref{opt_lambda_f_theorem} indicates that by properly choosing the value of $\lambda_f$, we can achieve the maximum network-wide secrecy throughput for the FD tier while guaranteeing a minimum required network-wide throughput for the HD tier.
Substituting the optimal $\lambda_f^*$ into \eqref{st_max} yields the maximum $\mathcal{T}_s^*$, which is shown in Fig. \ref{MTS_SIGMA_EP}.
Just as analyzed previously, only those $\sigma$ and $\epsilon$ that satisfy \eqref{sigma_epsilon} can yield a positive $\mathcal{T}^*_s$.
While $\mathcal{T}_{s}^*$ increases in $\epsilon$, $\mathcal{T}_{s}^*$ initially increases in $\sigma$ and then decreases in it.
The underlying reason is, too small a $\sigma$ corresponds to a small probability of successful transmission, whereas too large a $\sigma$ limits the transmission rate; either aspect results in a small $\mathcal{T}_{s}^*$, as can be seen from \eqref{st_def}.

\begin{figure}[!t]
	\centering
	\includegraphics[width=2.7in]{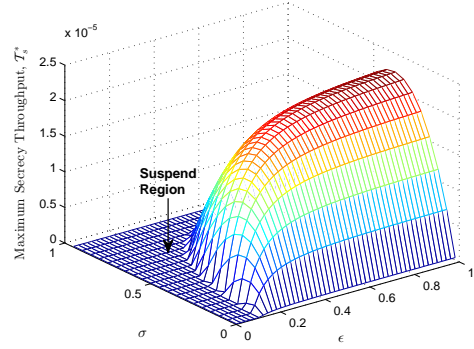}
	\caption{Maximum network-wide secrecy throughput vs. $\sigma$ and $\epsilon$, with $P_t = 20$dBm, $N_f=4$, $N_t=1$, $N_e=8$, $\lambda_e=10^{-2}$, $\sigma_c=0.9$ and $T_c=10^{-3}$.
		In the dark blue areas, there is no positive $\mathcal{T}_s$ that simultaneously satisfies connection and secrecy outage probabilities.}
	\label{MTS_SIGMA_EP}
\end{figure}

In addition to the density of the FD tier $\lambda_f$, the jamming power of the FD Rx $P_t$ \footnote{Since we only focus on the optimization of network density, the power control issue is outside the scope of this paper.} also triggers a non-trivial tradeoff between transmission reliability and secrecy, thus impacting network-wide secrecy throughput $\mathcal{T}_s$.
Similar to Theorem \ref{opt_lambda_f_theorem}, we can also prove the quasi-concavity of $\mathcal{T}_s$ on $P_t$, which is validated in Fig. \ref{TH_PT}.
We observe that, in a certain range of $\lambda_f$, how $\mathcal{T}_s$ varies w.r.t. $\lambda_f$ heavily depends on the value of $P_t$. 
For example, if at the small $P_t$ regime $\mathcal{T}_s$ increases in $\lambda_f$, $\mathcal{T}_s$ might decrease in $\lambda_f$ in the large $P_t$ regime.

\begin{figure}[!t]
	\centering
	\includegraphics[width=2.7in]{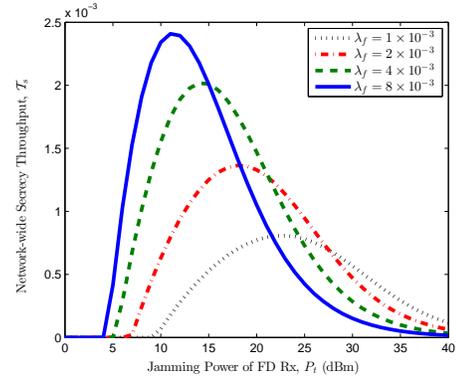}
	\caption{Network-wide secrecy throughput vs. $P_t$ for different values of $\lambda_f$ with $N_f=4$, $N_t=1$, $N_e=4$, $\lambda_e=10^{-3}$, $\sigma=0.9$, $\sigma_c=0.9$, $\epsilon=0.1$ and $T_c=10^{-3}$.}
	\label{TH_PT}
\end{figure}

\section{Multi-antenna-jamming FD Receiver}
In this section, we consider the scenario of the FD Rx using multi-antenna jamming.
Thanks to the extra degrees of freedom provided by multiple jamming antennas, each FD Rx is able to inject jamming signals into the null space of the SI channel such that SI will not leak out to the Rx's input, and MRC reception can be simply adopted at the input for the desired signal.
This is inspired by the idea of the artificial noise scheme proposed in \cite{Goel2008Guaranteeing}.
We will see the analysis of multi-antenna jamming is much more different from and much more sophisticated than that of single-antenna jamming.

Without lose of generality, we consider a typical FD Rx at the origin $o$.
The details of SSIC are given as follows.
We first use MRC reception at the input of the typical FD Rx, the weight vector of which can be obtained from \eqref{signal_fd}, i.e., $\tilde{\bm w}_{f} = \frac{\bm f^H_{\hat o o}}{\|\bm f_{\hat oo}\|}$.
Here we use superscript \~{} to distinguish the multi-antenna jamming case from the single-antenna jamming case.
We then design the jamming signal $\bm v_{o}$ in \eqref{signal_fd} in the form of $\bm v_{ o}=\tilde{\bm F}_{o}\tilde{\bm v}_{o}$, where $\tilde{\bm v}_{o}\in\mathbb{C}^{N_j\times 1}$ is an $N_j$-stream jamming signal vector with i.i.d. entries $\tilde v_i\sim\mathcal{CN}\left(0,{1}/{N_j}\right)$ and $N_j\le N_t-1$, $\tilde{\bm F}_{o}\in\mathbb{C}^{N_t\times N_j}$ is the projection matrix onto the null space of vector $\left(\tilde{\bm w}_{f} \bm F_{oo}\right)^H$ such that the columns of $\left[\frac{\left(\tilde{\bm w}_{f} \bm F_{oo}\right)^H}{\|\tilde{\bm w}_{f} \bm F_{oo}\|},\tilde{\bm F}_{o} \right]$ constitute an orthogonal basis, i.e., $\tilde{\bm w}_{f}\bm F_{o o}\bm v_{ o}=0$.
In this way, SI is completely eliminated in the spatial domain.
It is worth mentioning that, $\tilde{\bm v}_{o}$ includes but is not limited to an $N_t-1$-stream signal vector.
Although $N_t-1$-dimension null space should better be injected with jamming signals to confuse eavesdroppers in a point-to-point transmission \cite{Goel2008Guaranteeing}, there is no general conclusion from the network perspective, since jamming signals impair not only eavesdroppers but also legitimate users.

\subsection{Connection Probability}
From the above discussion, the SIR of the typical FD Rx can be obtained from \eqref{signal_fd}, which is
\begin{align}\label{sir_fd_nt}
\widetilde{\textsf{SIR}}_{f} = \frac{P_f\|\bm f_{\hat o o}\|^2D_f^{-\alpha}}
{\tilde I_h+\tilde I_f},
\end{align}
where $\tilde I_h\triangleq\sum_{\hat z\in\hat \Phi_h}{P_h |\tilde{\bm w}_{f} \bm f_{\hat z o}|^2}{D_{\hat z o}^{-\alpha}}$ and $\tilde I_f\triangleq\sum_{\hat z\in\hat \Phi_f\setminus \hat o}\left({P_f |\tilde{\bm w}_{f}\bm f_{\hat z o}|^2}{D_{\hat z o}^{-\alpha}}
+\left({P_t}/{N_j}\right){\|\tilde{\bm w}_{f} \bm F_{z o}\tilde{\bm F}_{z}\|^2}{D_{zo}^{-\alpha}}\right)$ are the aggregate interferences from the HD and FD tiers, respectively.
Substituting \eqref{sir_fd_nt} into \eqref{pt_fd_def} produces the connection probability of the typical FD Rx, denoted by $\tilde{\mathcal{P}}_{t}$.
As discussed in the single-antenna jamming case, the exact expression of $\tilde{\mathcal{P}}_{t}$ can be derived, which however is not in an analytical form.
Instead, we provide a more tractable lower bound for $\tilde{\mathcal{P}}_{t}$ in the following theorem.
\begin{theorem}\label{ptnt_bound_theorem}
	The connection probability $\tilde{\mathcal{P}}_{t}$ of the typical FD Rx is lower bounded by
	\begin{align}\label{ptnt_bound}
	\tilde{\mathcal{P}}_{t}^L = e^{-\tilde{\Lambda}_{f}^L\beta_t^{\delta}}\left(1
	+\sum_{m=1}^{N_f-N_t-1}\frac{1}{m!}\sum_{n=1}^{m}\left(
	\delta\tilde{\Lambda}_{f}^L\beta_t^{\delta}\right)^n
	\Upsilon_{m,n}\right),
	\end{align}
	where $\tilde{\Lambda}_{f}^L\triangleq
	C_{\alpha,2}P_{hf}^{\delta}\lambda_hD_f^2
	+C_{\alpha,2}\lambda_fD_f^2
	+C_{\alpha,N_j+1}\left({P_{tf}}/{N_j}\right)^{\delta}
	\lambda_fD_f^2$ and $\Upsilon_{m,n}$ has been defined in \eqref{pt_bound}.
\end{theorem}
\begin{proof}
	Please see Appendix \ref{appendix_ptnt_bound_theorem}.
\end{proof}

To further facilitate the analysis, an approximation for $\tilde{\mathcal{P}}_{t}$ is provided by the following corollary.
\begin{corollary}\label{ptnt_high_corollary}
	In the large probability region, i.e., $\tilde{\mathcal{P}}_{t}\rightarrow 1$, $\tilde{\mathcal{P}}_{t}$ is approximated by
	\begin{equation}\label{ptnt_high}
	\tilde{\mathcal{P}}_{t} \approx 1-\tilde{\Lambda}_{f}\beta_t^{\delta}
	K_{\alpha,N_f-N_t},
	\end{equation}
	where $\tilde{\Lambda}_{f}=\tilde{\Lambda}_{f}^L$ and $K_{\alpha,N}$ has been defined in Corollary \ref{pt_high_corollary}.
\end{corollary}

\begin{figure}[!t]
	\centering
	\includegraphics[width=2.8in]{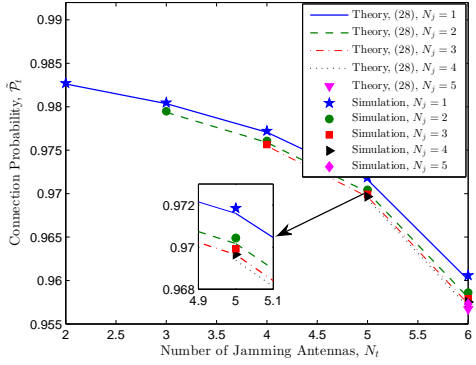}
	\caption{Connection probability vs. $N_t$ for different values of $N_j$, with $P_t = 20$dBm, $N_f=8$, $\lambda_f=10^{-3}$ and $\beta_t = 1$.}
	\label{PTNJ_NT_NJ}
\end{figure}

Fig. \ref{PTNJ_NT_NJ} shows that the result in \eqref{ptnt_high} approximates to the exact connection probability  $\tilde{\mathcal{P}}_{t}$ provided by Monte-Carlo simulations.
We see that $\tilde{\mathcal{P}}_{t}$ greatly reduces as the number $N_t$ of jamming antennas increases.
In addition, $\tilde{\mathcal{P}}_{t}$ suffers a slight decrease when the number $N_j$ of jamming signal streams increases.
This implies when the value of $N_t$ is fixed, $\tilde{\mathcal{P}}_{t}$ is less insensitive to the value of $N_j$.

Following similar steps in Theorem \ref{ptnt_bound_theorem} and Corollary \ref{ptnt_high_corollary}, an approximation for the connection probability $\tilde{\mathcal{P}}_{c}$ of a typical HD Rx in the large probability region is given in the following corollary.
\begin{corollary}\label{pcnt high_corollary}
	Define $\tilde{\Lambda}_{h}=C_{\alpha,2}\lambda_hD_h^2+
	C_{\alpha,2}P_{fh}^{\delta}\lambda_fD_h^2+C_{\alpha,N_j+1}
	\left({P_{th}}/{N_j}\right)^{\delta}
	\lambda_fD_h^2$.
	In the large probability region, i.e., $\tilde{\mathcal{P}}_{c} \rightarrow 1$, the connection probability $\tilde{\mathcal{P}}_{c}$ of a typical HD Rx is approximated by
	\begin{equation}\label{pcnt_high}
	\tilde{\mathcal{P}}_{c} \approx 1-\tilde{\Lambda}_{h}\beta_c^{\delta}
	K_{\alpha,N_h}.
	\end{equation}
\end{corollary}

\subsection{Secrecy Outage Probability}
Assuming that every eavesdropper has the capability of multiuser decoding and uses the MMSE receiver, the weight vector of the eavesdropper located at $e$ can be obtained from \eqref{signal_eve}, which is
\begin{equation}\label{weight_mmse_nt}
\tilde{\bm w}_{e} = \tilde{\bm R}_{e} ^{-1}\bm g_{\hat oe},
\end{equation}
where $\tilde{\bm R}_{e}\triangleq ({P_t}/{N_j})\bm G_{oe}
\tilde{\bm F}_{o}\tilde{\bm F}_{o}^H\bm G_{ oe}^HD_{o e}^{-\alpha}
+\sum_{z\in\Phi_f\setminus o}({P_t}/{N_j})\bm G_{ze}
\tilde{\bm F}_{z}\tilde{\bm F}_{z}^H\bm G_{ ze}^HD_{z e}^{-\alpha}$;
the resulting SIR is
\begin{equation}\label{sir_e_nt}
\widetilde{\textsf{SIR}}_{e} = P_f\bm g_{\hat oe}^H\tilde{\bm R}_{e}^{-1}\bm g_{\hat oe}D_{\hat o e}^{-\alpha}.
\end{equation}

Due to the existence of multi-stream jamming signals, computing secrecy outage probability requires using the integer partitions of a positive integer \cite{Louie2011Spatial}.
For convenience, we describe the integer partitions of a positive integer $k$ by introducing an integer partition matrix $\bm\Theta_k$.
For example, the integer partitions of $4$ is characterized by
\begin{align}
\small{\mathbf{\Theta}_4 = \left[\begin{array}{ccccc}
	4 &  &  &  & \\
	3 & 1 &  &  & \\
	2 & 2 &  &  &   \\
	2 & 1 & 1 &  &   \\
	1 & 1 & 1 & 1 & 1
	\end{array}\right]}.
\label{QM}
\end{align}
In the following, we denote $|\xi_k|$ as the number of rows of $\bm\Theta_k$, and $\xi_{i,j,k}$, $|\xi_{j,k}|$, $\phi_{i,j,k}$ and $|\phi_{j,k}|$ as the $i$-th
entry, the number of entries, the number of the
$i$-th largest entry and the number of non-repeated entries in the $j$-th row of $\bm\Theta_k$, respectively.
We provide a closed-form expression for secrecy outage probability $\tilde{\mathcal{P}}_{so}$ in the following theorem.
\begin{theorem}\label{psnt_approx_theorem}
    The secrecy outage probability of a typical multi-antenna jamming FD Rx is
    \begin{align}\label{psnt_approx}
        &\tilde{\mathcal{P}}_{so}= 1 - \exp\Bigg(
        -\frac{\pi\lambda_e}{{C_{\alpha,N_j+1}\lambda_f}}
        \sum_{n=0}^{N_e-1}
        \sum_{i=0}^{\min(n,N_j)}\binom{N_j}{i}
\times\nonumber\\
        &\frac{\left({P_{tf}\beta_e}/{N_j}\right)^{i-\delta}}{\left(1+{P_{tf}\beta_e}/{N_j}\right)^{N_j}}{}
        \sum_{j=1}^{|\xi_{n-i}|}
{(-1)^{|\xi_{j,n-i}|}|\xi_{j,n-i}|!\Xi_{j,n-i}}
        \Bigg),
    \end{align}
where $\Xi_{j,n}=\frac{\prod_{m=1}^{|\xi_{j,n}|}
\prod_{k=1}^{\xi_{m,j,n}}\frac{(N_j+1-k)(k-1-\delta)}
{k(N_j-k+\delta)}}{\prod_{i=1}^{|\phi_{j,n}|}\phi_{i,j,n}!}$.
Here, we let $|\xi_0|=1$, $|\xi_{j,0}|=0$ and $\Xi_{j,0}=1$.
\end{theorem}
\begin{proof}
    Please see Appendix \ref{appendix_psnt_approx_theorem}.
\end{proof}

We see that secrecy outage probability $\tilde{\mathcal{P}}_{so}$ is affected by the number $N_j$ of jamming signal streams rather than the number $N_t$ of jamming antennas.
The result in \eqref{psnt_approx} is well verified by Monte-Carlo simulations, just as shown in Fig. \ref{PSO_NJ_NE_LF}.
To better understand the effect of $N_j$ on $\tilde{\mathcal{P}}_{so}$, we investigate the asymptotic behavior of $\tilde{\mathcal{P}}_{so}$ w.r.t. $N_j$ by considering the cases $N_j = 1$ and $N_j\rightarrow\infty$, respectively.
\begin{figure}[!t]	
	\centering
	\includegraphics[width=2.8in]{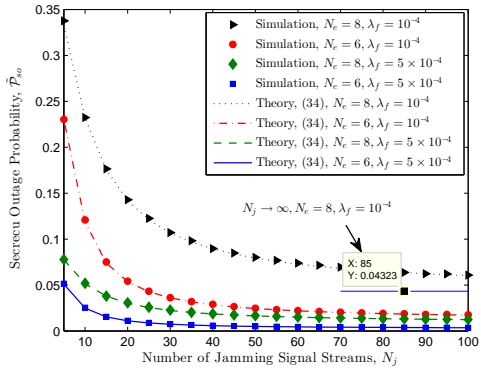}
	\caption{Secrecy outage probability vs. $N_j$ for different values of $N_e$ and $\lambda_f$, with $P_t = 10$dBm, $\lambda_e=10^{-4}$ and $\beta_e = 1$.}
	\label{PSO_NJ_NE_LF}
\end{figure}
\begin{corollary}\label{psnt_corollary1}
	When $N_j=1$, $\tilde{\mathcal{P}}_{so}$ in \eqref{psnt_approx} shares the same expression as the one given in \eqref{ps_approx}.
\end{corollary}
\begin{proof}
	Substituting $N_j=1$ into $\Xi_{j,n-i}$, we have $\Xi_{j,n-i}=0$ for $j<|\xi_{n-i}|$.
	Since the $|\xi_{n-i}|$-th integer partition of $n-i$ (i.e., the last row of $\mathbf{\Theta}_{n-i}$) must be $n-i$ ones, we have $|\phi_{|\xi_{n-i}|,n-i}|$=1 and $|\xi_{|\xi_{n-i}|,n-i}|=\phi_{1,|\xi_{n-i}|,n-i}$.
	Therefore, the term $\sum_{j=1}^{|\xi_{n-i}|}
	{(-1)^{|\xi_{j,n-i}|}|\xi_{j,n-i}|!\Xi_{j,n-i}}$ in \eqref{psnt_approx} reduces to ${(-1)^{|\xi_{|\xi_{n-i}|,n-i}|}|\xi_{|\xi_{n-i}|,n-i}|!
		\Xi_{|\xi_{n-i}|,n-i}}=1$,
	substituting which into \eqref{psnt_approx} completes the proof.
\end{proof}

Corollary \ref{psnt_corollary1} implies emitting a single-stream jamming signal using multiple antennas has the same effect as single-antenna jamming in confounding eavesdroppers.

\begin{corollary}\label{psnt_corollary2}
    As $N_j\rightarrow \infty$, $\tilde{\mathcal{P}}_{so}$ in \eqref{psnt_approx} tends to the following constant value
    \begin{align}\label{psnt_infinity}
1 - \exp\Bigg(&
-\frac{\lambda_e}{{\Gamma(1-\delta)\lambda_f}}    \sum_{n=0}^{N_e-1}
\sum_{i=0}^{n}\frac{
e^{-P_{tf}\beta_e}}{i!(P_{tf}\beta_e)^{\delta-i}}\times\nonumber\\
&     \sum_{j=1}^{|\xi_{n-i}|}
{\frac{|\xi_{j,n-i}|!\prod_{m=1}^{|\xi_{j,n-i}|}\prod_{k=1}^{\xi_{m,j,n-i}}\frac{k-1-\delta}{k}}
	{(-1)^{|\xi_{j,n-i}|}\prod_{i=1}^{|\phi_{j,n-i}|}\phi_{i,j,n-i}!}}\Bigg).
    \end{align}
\end{corollary}
\begin{proof}
    Invoking $\lim_{N\rightarrow\infty}\frac{\Gamma(N+\delta)}
    {\Gamma(N)N^{\delta}}=1$ and $\lim_{N\rightarrow\infty}\left(1+\frac{x}{N}\right)^N
    =e^{x}$ in \eqref{psnt_approx} directly yields \eqref{psnt_infinity}.
\end{proof}

Corollary \ref{psnt_corollary2} implies increasing jamming signal streams can not arbitrarily reduce secrecy outage probability, just as validated in Fig. \ref{PSO_NJ_NE_LF}.
This is because the total power $P_t$ of jamming signals is limited.
Conversely, if $P_{t}$ in \eqref{psnt_infinity} goes to infinity, $\tilde{\mathcal{P}}_{so}$ reduces to zero.

\subsection{Network-wide Secrecy Throughput}
The network-wide secrecy throughput $\tilde{\mathcal{T}}_{s}$ in multi-antenna jamming scenario under a connection probability $\tilde{\mathcal{P}}_{t}(\beta_{t})=\sigma$ and a secrecy outage probability $\tilde{\mathcal{P}}_{so}(\beta_{e})=\epsilon$ has the same form as \eqref{st_def}.
We aim to optimize $\lambda_f$ to maximize $\tilde{\mathcal{T}}_{s}$ while guaranteeing a certain level of the HD tier throughput, i.e., $\tilde{\mathcal{T}}_{c}\ge {T}_c$, with $\tilde{\mathcal{T}}_{c}$ sharing the form of ${\mathcal{T}}_{c}$ given in Sec. II-A.
 Before proceeding to the optimization problem, we compute $\beta_t^*$ and $\beta_e^*$ from the equations $\tilde{\mathcal{P}}_{t}(\beta_{t})=\sigma$ and $\tilde{\mathcal{P}}_{so}(\beta_{e})=\epsilon$, respectively.
\begin{proposition}\label{beta_t_nt_lemma}
    In the large probability region, i.e., $\sigma\rightarrow 1$, the root of $\tilde{\mathcal{P}}_{t}(\beta_{t})=\sigma$ is given by
    \begin{equation}\label{beta_t_nt}
\beta_t^* = \left(\frac{\left(1-\sigma\right)/\left(C_{\alpha,2}D_f^2 K_{\alpha,N_f-N_t}\right)}
{
P_{hf}^{\delta}\lambda_h
+\left(1
+\frac{C_{\alpha,N_j+1}}{C_{\alpha,2}}\left(\frac{P_{tf}}{N_j}\right)^{\delta}\right)
\lambda_f}\right)^{{\alpha}/{2}}.
    \end{equation}
\end{proposition}
\begin{proof}
    Recalling Corollary \ref{ptnt_high_corollary}, \eqref{beta_t_nt} is obtained by solving $1-\tilde{\Lambda}_{f}\beta_t^{\delta}
    K_{\alpha,N_f-N_t}=\sigma$.
\end{proof}

Generally, it is impossible to derive a closed-form expression for $\beta_e^*$ due to the complicated integer partitions in \eqref{psnt_approx}.
However, an analytical approximation for $\beta_e^*$ can be readily obtained from \eqref{ps_approx} by considering the single jamming signal stream, i.e., $N_j=1$, in the large-antenna eavesdropper case.
\begin{proposition}\label{beta_e_nt_lemma}
    When $N_e\gg 1$ and $N_j=1$, $\beta_e$ that satisfies $\tilde{\mathcal{P}}_{so}(\beta_{e})=\epsilon$ has the same expression as the one given in \eqref{beta_e}.
\end{proposition}

For the more general case $N_j\ge 2$, the value of $\beta_e^*$ can be obtained via numerical calculation, i.e., $\beta_e^* = \tilde{\mathcal{P}}_{so}^{-1}(\epsilon)$,
where $\tilde{\mathcal{P}}_{so}^{-1}(\epsilon)$ is the inverse function of $\tilde{\mathcal{P}}_{so}(\beta_e)$.

In Fig. \ref{TS_LF_NJ}, we illustrate some numerical examples of network-wide secrecy throughput $\tilde{\mathcal{T}}_{s}$.
We see that $\tilde{\mathcal{T}}_{s}$ first increases and then decreases as the FD tier density $\lambda_f$ increases.
The value of $\lambda_f$ should be properly chosen in order to maximize $\tilde{\mathcal{T}}_{s}$.
We also find that, $\tilde{\mathcal{T}}_{s}$ improves as the number $N_j$ of jamming signal streams increases on the premise of a fixed number $N_t$ of jamming antennas.
\begin{figure}[!t]
\centering
\includegraphics[width=2.5in]{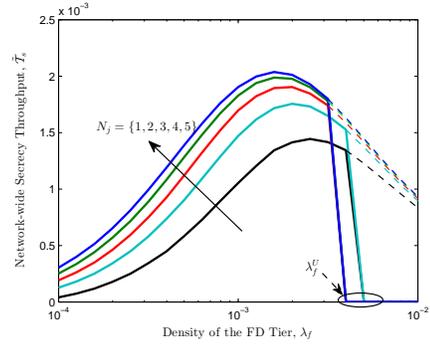}
\caption{Network-wide secrecy throughput vs. $\lambda_f$ for different values of $N_j$, with $P_t = 20$dBm, $N_f=N_e=8$, $N_t = 6$, $\lambda_e=10^{-4}$, $\sigma=\sigma_c=0.9$, $\epsilon=0.02$ and $T_c=10^{-3}$.
The dashed lines show the values of $\tilde{\mathcal{T}}_{s}$ without an HD tier throughput constraint, i.e. $T_c=0$.}
\label{TS_LF_NJ}
\end{figure}

Next, we formulate the problem of maximizing network-wide secrecy throughput $\tilde{\mathcal{T}}_{s}$ as follows,
\begin{subequations}
\begin{align}\label{stnt_max}
&\max_{\lambda_f}~ \tilde{\mathcal{T}}_{s}
=\frac{\lambda_f\sigma}{\ln2}
\left[\ln\frac{1+\tilde X(1+\tilde Y\lambda_f)^{-\alpha/2}}
{1+\tilde{\mathcal{P}}_{so}^{-1}(\epsilon)}\right]^+,\\
&\quad \mathrm{s.t.} ~0<\lambda_f\le \tilde\lambda_f^U,
\end{align}
\end{subequations}
where $\tilde\lambda_f^U
\triangleq\frac{({1-\sigma_c})/(D_h^2K_{\alpha,N_h})
\left(2^{{T_c}/({\lambda_h\sigma_c})}-1\right)
^{-\delta}-C_{\alpha,2}\lambda_h}{C_{\alpha,2}P_{fh}^{\delta}+C_{\alpha,N_j+1}\left({P_{th}}/{N_j}\right)^{\delta}}$ is obtained from $\tilde{\mathcal{T}}_{c}=T_c$, $\tilde X\triangleq\left(\frac{1-\sigma}{C_{\alpha,2}D_f^2
K_{\alpha,N_f-N_t}P_{hf}^{\delta}\lambda_h}\right)
^{\frac{\alpha}{2}}$ and $\tilde Y\triangleq\frac{C_{\alpha,2}+C_{\alpha,N_j+1}(P_{tf}/N_j)^{\delta}}
{C_{\alpha,2}P_{hf}^{\delta}\lambda_h}$.

For the single jamming signal stream case $N_j=1$, $\tilde{\mathcal{P}}_{so}^{-1}(\epsilon)$ has a closed-from expression given in \eqref{beta_e}, i.e., $\tilde{\mathcal{P}}_{so}^{-1}(\epsilon)=\beta_e^*
=\tilde Z \lambda_f^{-\frac{\alpha}{2}}$ with $\tilde Z\triangleq\frac{1}{P_{tf}}\left(
\frac{\pi\lambda_eN_e}{C_{\alpha,2}
	\ln\frac{1}{1-\epsilon}}\right)
^{\frac{\alpha}{2}}$.
Accordingly, problem \eqref{stnt_max} has the same form as problem \eqref{F_max}.
As a consequence, the optimal $\lambda_f$ that maximizes $\tilde{\mathcal{T}}_{s}$ also shares the same form as \eqref{opt_lambda_f}, simply with $X$, $Y$, $Z$ and $\lambda_f^U$ replaced by $\tilde X$, $\tilde Y$, $\tilde Z$ and $\tilde\lambda_f^U$, respectively.

For the more general case $N_j\ge 2$, we can only solve problem \eqref{stnt_max} using one-dimension exhaustive search in the range $(0, \tilde\lambda_f^U]$.
Since increasing the number $N_j$ of jamming signal streams always benefits network-wide secrecy throughput, we should set $N_j=N_t-1$.
Thus, $N_t-1$-dimension null space is fully injected with jamming signals.
In Fig. \ref{OPT_LAMBDA} and Fig. \ref{MAX_TS}, we illustrate the optimal density $\lambda_f^*$ and the corresponding maximum network-wide secrecy throughput $\tilde{\mathcal{T}}_{s}^*$,
respectively.

From Fig. \ref{OPT_LAMBDA}, we observe a general trend that the value of $\lambda_f^*$ decreases as $N_t$ increases on the premise of the existence of a positive $\tilde{\mathcal{T}}_{s}^*$.
The reason behind is twofold:
on one hand, adding jamming antennas provides relief to deploying more FD jammers to degrade the wiretap channels;
on the other hand, reducing the number of FD Tx-Rx pairs reduces network interference, thus improving the main channels.
How the value of $\lambda_f^*$ is influenced by $N_e$ depends on the specific values of $\lambda_e$ and $N_t$.
For example, if each eavesdropper adds receive antennas, more FD jammers are needed for a relatively small $N_t$ or a small $\lambda_e$ (see (a), (b) and (c)), whereas fewer FD jammers might be better as $N_t$ or $\lambda_e$ goes large (see $N_t=7$ in (c) and $N_t=5$ in (d)).
This is because, if we continue to add FD jammers, we can scarcely achieve a positive secrecy throughput .

\begin{figure}[!t]
	\centering
	\includegraphics[width=2.8in]{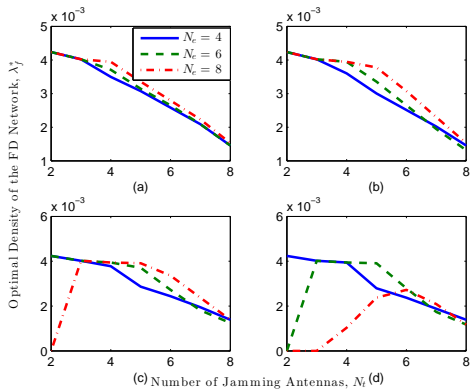}
	\caption{Optimal density of the FD tier vs. $N_t$ for different values of $N_e$ and (a) $\lambda_e=2\times10^{-4}$, (b) $\lambda_e=4\times10^{-4}$, (c) $\lambda_e=7\times10^{-4}$, (d) $\lambda_e=10^{-3}$, with $P_t = 20$dBm, $N_f=8$, $N_j=N_t-1$, $\lambda_f=2\times10^{-3}$, $\sigma=\sigma_c=0.9$, $\epsilon=0.02$ and $T_c=10^{-3}$.
		Note that $\lambda_f^*=0$ when $N_t=2$ in (c) and $N_t=2,3$ in (d).
		This means a positive secrecy throughput that simultaneously satisfies the connection and secrecy outage probability constraints can not be achieved, regardless of the value of $\lambda_f$.}
	\label{OPT_LAMBDA}
\end{figure}

\begin{figure}[!t]
	\centering
	\includegraphics[width=2.8in]{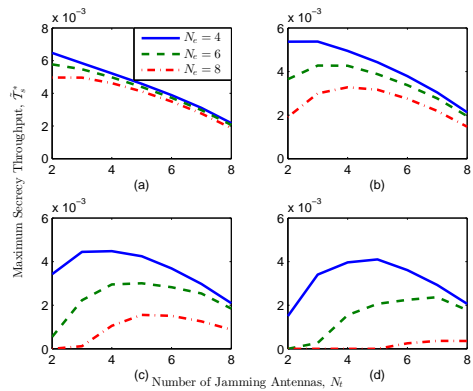}
	\caption{Maximum network-wide secrecy throughput vs. $N_t$ for different values of $N_e$ and (a) $\lambda_e=2\times10^{-4}$, (b) $\lambda_e=4\times10^{-4}$, (c) $\lambda_e=7\times10^{-4}$, (d) $\lambda_e=10^{-3}$, with $P_t = 20$dBm, $N_f=8$, $N_j=N_t-1$, $\lambda_f=2\times10^{-3}$, $\sigma=\sigma_c=0.9$, $\epsilon=0.02$ and $T_c=10^{-3}$.}
	\label{MAX_TS}
\end{figure}

In Fig. \ref{MAX_TS}, we see that the maximum network-wide secrecy throughput $\tilde{\mathcal{T}}_{s}^*$ always deteriorates as $\lambda_e$ or $N_e$ increases.
How the value of $\tilde{\mathcal{T}}_{s}^*$ is affected by $N_t$ depends on the specific values of $\lambda_e$ and $N_e$.
Specifically, for relatively small values of $\lambda_e$ and $N_e$, $\tilde{\mathcal{T}}_{s}^*$ decreases as $N_t$ increases (see (a)).
This means we should use as few jamming antennas as possible.
However, as $\lambda_e$ or $N_e$ increases, $\tilde{\mathcal{T}}_{s}^*$ first increases and then decreases as $N_t$ increases (see (b), (c) and (d)).
This implies that a modest value of $N_t$ is required to balance improving the main channels with degrading the wiretap channels.

\section{Conclusions}
This paper comprehensively studies physical-layer security using FD Rx jamming techniques against randomly located eavesdropper in a heterogeneous DWN consisting of both HD and FD tiers.
The connection probability and the secrecy outage probability of a typical FD Rx is analyzed for single- and multi-antenna jamming scenarios, and the optimal FD tier density is provided for maximizing network-wide secrecy throughput under constraints including the given dual probabilities and the network-wide throughput of the HD tier.
Numerical results are presented to validate our theoretical analysis, and show the benefits of FD Rx jamming in improving network-wide secrecy throughput.

\appendix
\subsection{Proof of Theorem \ref{pt_exact_theorem}}
\label{appendix_pt_exact_theorem}
Let $s\triangleq {D_f^{\alpha}\beta_t}/{P_f}$ and $I_o = I_h+I_f$.
$\mathcal{P}_t$ can be calculated by substituting \eqref{sir_fd} into \eqref{pt_fd_def}
\begin{align}\label{pt_k1}
&\mathcal{P}_{t}= \mathbb{E}_{I_o} \left[ \mathbb{P}\left\{\|\bm f_{ \hat o o}^H \bm U\|^2\geq sI_o\right\}\right]
    \stackrel{\mathrm{(a)}}
    = \mathbb{E}_{I_o}\left[e^{-sI_o}
    \sum_{m=0}^{N_f-3}\frac{s^mI_o^m}
    {m!}\right]\nonumber\\
    & =\sum_{m=0}^{N_f-3}\mathbb{E}_{I_o} \left[
    \frac{s^me^{-sI_o}}
    {m!}I_o^m\right]
    \stackrel{\mathrm{(b)}}= \sum_{m=0}^{N_f-3}\left[
     \frac{s^m\mathcal{L}^{(m)}_{I_o}(s)}
    {(-1)^mm!}\right],
\end{align}
where (a) holds for $\|\bm f_{ \hat o o}^H \bm U\|^2\sim \Gamma(N_f-2,1)$, and (b) is obtained from \cite[Theorem 1]{Hunter08Transmission}.
Due to the independence of $I_{h}$ and $I_{f}$, $\mathcal{L}_{I_o}(s)$ is given by
\begin{align}\label{la_io}
    \mathcal{L}_{I_o}(s) = \mathbb{E}_{I_o}\left[e^{-sI_o}\right]
    =\mathcal{L}_{I_h}(s)\mathcal{L}_{I_f}(s).
\end{align}
$\mathcal{L}_{I_h}(s)$ can be directly obtained from \cite[(8)]{Haenggi2009Stochastic}, which is
\begin{equation}\label{la_hd}
  \mathcal{L}_{I_h}(s)
=\exp\left(-\lambda_hC_{\alpha,2}(P_hs)^{\delta}\right).
\end{equation}
$\mathcal{L}_{I_f}(s)= \mathbb{E}_{I_f}\left[e^{-sI_f}\right]$ can be computed as
\begin{align}
&\mathcal{L}_{I_f}(s)
\label{la_fd1}
=\mathbb{E}_{\hat \Phi_f}\left[
\prod_{\hat z\in\hat \Phi_f\setminus \hat o}e^{-s\left(\frac{P_f |\bm w_{f} \bm f_{\hat z o}|^2}{D_{\hat z o}^{\alpha}} +\frac{P_t|\bm w_{f} \bm f_{z o}|^2}{D_{zo}^{\alpha}}\right)}\right]\\
&\label{la_fd2}
\stackrel{\mathrm{(c)}}=
\mathbb{E}_{\hat \Phi_f}
\left[\prod_{ \hat z\in\hat \Phi_f\setminus \hat o}
\frac{1}{1+P_fsD_{\hat z o}^{-\alpha}}\frac{1}{1+P_tsD_{zo}^{-\alpha}}\right]\\
&\label{la_fd}
\stackrel{\mathrm{(d)}}=\exp\Bigg(-\lambda_f
\int_0^{\infty}\int_0^{2\pi}\Bigg(1-\frac{1}{1+P_fsr^{-\alpha}}\times\nonumber\\
&\qquad\frac{1}{1+P_ts\left(r^2+D_f^2-2rD_f \cos\theta\right)^{-\alpha/2}}\Bigg)rd\theta dr
\Bigg),
\end{align}
where (c) holds for $|\bm w_{f} \bm f_{\hat zo}|^2, |\bm w_{f} \bm f_{z o}|^2\sim \mathrm{Exp}(1)$, and (d) is derived by using the probability generating functional (PGFL) over a PPP \cite{Stoyan1996Stochastic}.
Substituting \eqref{la_hd} and \eqref{la_fd} into \eqref{la_io} completes the proof.

\subsection{Proof of Theorem \ref{pt_bound_theorem}}
\label{appendix_pt_bound_theorem}

To provide a lower bound for $\mathcal{P}_t$, we needs only provide a lower bound for $\mathcal{L}_{I_f}(s)$.
This is because, a lower bound for $\mathcal{L}_{I_f}(s)$ actually overestimates the aggregate interference $I_f$ from the FD tier, which leads to a lower bound for $\mathcal{P}_t$.
From \eqref{la_fd1}, we have
\begin{align}\label{la_fd_low}
&\mathcal{L}_{I_f}(s)\nonumber\\
&\stackrel{\mathrm{(e)}}\ge
\mathbb{E}_{\hat \Phi_f}\left[
\prod_{\hat  z\in\hat \Phi_f\setminus\hat  o}e^{-\frac{sP_f |\bm w_{f} \bm f_{\hat zo}|^2}{D_{\hat z o}^{\alpha}} }\right] \mathbb{E}_{\Phi_f}\left[
\prod_{z\in \Phi_f\setminus o}e^{-\frac{sP_t|\bm w_{f} \bm f_{zo}|^2}{D_{z o}^{\alpha}}}\right]\nonumber\\ &\stackrel{\mathrm{(f)}} =
e^{-\lambda_fC_{\alpha,2}(P_fs)^{\delta}}
e^{-\lambda_fC_{\alpha,2}(P_ts)^{\delta}},
\end{align}
where (e) follows from the FKG inequality \cite[Theorem 10.13]{Haenggi2012Stochastic}, since both $\prod_{\hat  z\in\hat \Phi_f\setminus \hat o}e^{-sP_f |\bm w_{f} \bm f_{\hat z o}|^2 D_{\hat z o}^{-\alpha}}$ and $\prod_{\in \Phi_f\setminus o}e^{-sP_t|\bm w_{f} \bm f_{zo}|^2D_{zo}^{-\alpha}}$ are decreasing random variables as the number of terms increases;
(f) holds for realizing that both $\hat \Phi_f\setminus \hat o$ and $\Phi_f\setminus o$ are PPPs with the same density $\lambda_f$ due to the displacement theorem \cite[page 35]{Haenggi2012Stochastic} and invoking \cite[(8)]{Haenggi2009Stochastic}.

Substituting \eqref{la_hd} and \eqref{la_fd_low} into \eqref{pt_exact} and invoking \cite[Theorem 1]{Hunter08Transmission}, we obtain the lower bound $\mathcal{P}_{t}^L$.

An upper bound for $\mathcal{P}_t$ can be obtained by calculating an upper bound for $\mathcal{L}_{I_f}$.
From \eqref{la_fd2}, we have
\begin{align}\label{la_fd_up}
&\mathcal{L}_{I_f}(s)
\stackrel{\mathrm{(g)}}\le\nonumber\\
&\mathbb{E}_{\hat \Phi_f}\!\!\left[
\prod_{ \hat z\in\hat \Phi_f\setminus \hat o}
\frac{1}{\left(1+P_fsD_{\hat z o}^{-\alpha}\right)^2}\right]\!
\mathbb{E}_{\Phi_f}\!\!\left[
\prod_{z\in \Phi_f\setminus o}
\frac{1}{\left(1+P_tsD_{z o}^{-\alpha}\right)^2}\right]
\nonumber\\
&\stackrel{\mathrm{(h)}} =
\exp\left(-\lambda_fC_{\alpha,2}\frac{1+\delta}{2}(P_fs)^{\delta}(1+P_{tf}^{\delta})\right),
\end{align}
where (g) follows from the Cauchy-Schwarz inequality and (h) holds for the PGFL over a PPP.
Substituting \eqref{la_hd} and \eqref{la_fd_up} into \eqref{pt_exact} and invoking \cite[Theorem 1]{Hunter08Transmission}, we obtain the upper bound $\mathcal{P}_{t}^U$.

\subsection{Proof of Theorem \ref{ps_exact_theorem}}
\label{appendix_ps_exact_theorem}
Let $r\triangleq D_{\hat oe}$.
Substituting \eqref{sir_e} into \eqref{ps_fd_def} and applying the PGFL over a PPP yield
    \begin{equation}\label{ps_app1}
   \mathcal{P}_{so} =
    1- \exp\left(-\lambda_e
    \int_0^{\infty}\int_0^{2\pi}\mathbb{P}
    \left\{\textsf{SIR}_{e}
    \ge\beta_{e}\right\}rd\theta dr\right). \!
\end{equation}
 Define $v\triangleq { r^{\alpha}\beta_e}/{P_f}$, $\mathbb{P}
 \left\{\textsf{SIR}_{e}
    \ge\beta_{e}\right\}$ in \eqref{ps_app1} can be calculated by invoking \cite[(11)]{Gao1998Theoretical}, i.e.,
\begin{equation}\label{pro_sir}
  \mathbb{P}\left\{\textsf{SIR}_{e}
    \ge\beta_{e}\right\}=
    \mathbb{E}_{\Phi_f}\left[
    \frac{1}{W}\sum_{n=0}^{N_e-1}w_nv^n\right],
\end{equation}
where $W=\left(1+P_tD_{oe}^{-\alpha}v\right)\prod_{ z\in \Phi_f\setminus o}\left(1+P_tD_{ ze}^{-\alpha}v\right)$, and $w_n$ is the coefficient of $v^n$ in $W$, which is
\begin{equation}\label{w_n}
  w_n = \sum_{i=0}^{\min(n,1)}\frac{\left(P_tD_{\hat oe}^{-\alpha}\right)^i}{(n-i)!}
  \sum_{
  z_1,\cdots,z_{n-i}
  \in\Phi_f\setminus o}
  \prod_{j=1}^{n-i}\frac{P_t}{D_{z_je}^{\alpha}}.
\end{equation}
Substituting $W$ and $w_n$ into \eqref{pro_sir}, we have
\begin{align}\label{pro_sir_2}
    &\mathbb{P}\left\{\textsf{SIR}_{e}
    \ge\beta_{e}\right\}
=  \mathbb{E}_{\Phi_f}\Bigg[
    \sum_{n=0}^{N_e-1}\sum_{i=0}^{\min(n,1)}
    \frac{1}{(n-i)!}\times\nonumber\\
& \quad \frac{\left(P_tD_{ oe}^{-\alpha}v\right)^i}{\left(1+P_tD_{ oe}^{-\alpha}v\right)}  \sum_{z_1,\cdots,z_{n-i}
  \in\Phi_f\setminus o}\frac{P_t^{n-i}v^{n-i}
  \prod_{j=1}^{n-i}D_{z_je}^{-\alpha}}
    {\prod_{z\in\Phi_f\setminus o}\left(1+P_tD_{ ze}^{-\alpha}v\right)}\Bigg]\nonumber\\
    &=\sum_{n=0}^{N_e-1}\sum_{i=0}^{\min(n,1)}
    \frac{\left(P_tD_{ oe}^{-\alpha}v\right)^i}{\left(1+P_tD_{ oe}^{-\alpha}v\right)(n-i)!}\times\nonumber\\
&   \qquad \mathbb{E}_{\Phi_f}\left[\sum_{z_1,\cdots, z_{n-i}
  \in\Phi_f\setminus o}\frac{P_t^{n-i}v^{n-i}
  \prod_{j=1}^{n-i}D_{z_je}^{-\alpha}}
    {\prod_{z\in\Phi_f\setminus o}\left(1+P_tD_{ ze}^{-\alpha}v\right)}\right],
\end{align}
where the expectation term can be calculated by using Campbell-Mecke theorem \cite[Theorem 4.2]{Stoyan1996Stochastic}
\begin{align}\label{campbell}
\mathbb{E}_{\Phi_f}&\left[\sum_{ z_1,\cdots, z_{n-i}
  \in\Phi_f\setminus o}\frac{P_t^{n-i}v^{n-i}
  \prod_{j=1}^{n-i}D_{z_je}^{-\alpha}}
    {\prod_{z\in\Phi_f\setminus o}\left(1+P_tD_{ ze}^{-\alpha}v\right)}\right]\nonumber\\
    &
    =\left(2\pi\lambda_f\int_0^{\infty}
    \frac{P_t v r^{-\alpha}}{1+P_t v r^{-\alpha}}rdr\right)^{n-i}\times\nonumber\\
    &\qquad
    \exp\left(
    -2\pi\lambda_f\int_0^{\infty}
    \frac{P_t v r^{-\alpha}}{1+P_t v r^{-\alpha}}rdr\right)\nonumber\\
    &
    \stackrel{\mathrm{(i)}} = \left(C_{\alpha,2}\lambda_fP_t^{\delta}v
    ^{\delta}\right)^{n-i}\exp\left(
    -C_{\alpha,2}\lambda_fP_t^{\delta}v
    ^{\delta}\right),
\end{align}
where (i) is obtained by transforming $P_t v r^{-\alpha}\rightarrow \mu$ and invoking formula \cite[(3.241.2)]{Gradshteyn2007Table}.
Substituting \eqref{pro_sir_2} and \eqref{campbell} into \eqref{ps_app1} and using $D_{\hat oe}=\sqrt{D_{oe}^2+D_f^2-2D_{oe}D_f\cos{\theta_o}}$, we complete the proof.

\subsection{Proof of Corollary \ref{ps_approx_corollary}}
\label{appendix_ps_approx_corollary}
As $D_f\rightarrow 0$, $\mathcal{Q}_i(r)={2\pi\left(P_{tf}\beta_e\right)^i}/
    ({1+P_{tf}\beta_e})$.
    Let $A\triangleq C_{\alpha,2}\lambda_f
        \left(P_{tf}\beta_e\right)^{\delta}$.
        Substituting $\mathcal{Q}_i(r)$ into \eqref{ps_exact} yields
\begin{align}\label{ps_approx_1}
&\mathcal{P}_{so}= 1 - \exp\Bigg(
-\lambda_e\sum_{n=0}^{N_e-1}\sum_{i=0}^{\min(n,1)}
\frac{A^{n-i}}
{(n-i)!}
\times\nonumber\\
&\qquad \qquad \qquad \qquad \frac{2\pi\left(P_{tf}\beta_e\right)^i}
{1+P_{tf}\beta_e}\int_0^{\infty}r^{2(n-i)}
    e^{-Ar^2}rdr\Bigg)\nonumber\\
&\stackrel{\mathrm{(j)}}=1 -
\exp\left(
-\lambda_e\sum_{n=0}^{N_e-1}
\sum_{i=0}^{\min(n,1)}
\frac{\pi}
{A}
\frac{\left(P_{tf}\beta_e\right)^i}
{1+P_{tf}\beta_e}\right)\nonumber\\
& = 1 - \exp\left(-\frac{\pi\lambda_e}{A(1+P_{tf}\beta_e)}
\sum_{n=0}^{N_e-1}
(1+P_{tf}\beta_e)^n\right),
\end{align}
where (j) holds for $\int_0^{\infty}r^{2(n-i)}e^{-Ar}rdr=\frac{(n-i)!}
{A^{n-i+1}}$.
Substituting $A$ into \eqref{ps_approx_1} completes the proof.

\subsection{Proof of Theorem \ref{opt_lambda_f_theorem}}
\label{appendix_opt_lambda_f_theorem}
To complete the proof, we need only derive the optimal $\lambda_f$, denoted by $\lambda_f^{\star}$, that maximizes $F(\lambda_f)$ in the range $[\lambda_f^L,\infty)$.
Clearly, if $0<\lambda_f^L\le\lambda_f^U$, the solution to problem \eqref{F_max} is $\lambda_f^* = \min(\lambda_f^{\star},\lambda_f^U)$; otherwise, there is no feasible solution.
    For convenience, we omit subscript $f$ from $\lambda_f$.
    Define $f_1(\lambda)= 1+X(1+Y\lambda)^{-\frac{\alpha}{2}}>1$, $f_2(\lambda)= 1+Z\lambda^{-\frac{\alpha}{2}}>1$ and $f(\lambda)= \ln\frac{f_1(\lambda)}{f_2(\lambda)}$, then the objective function in \eqref{F_max} changes into $F(\lambda)=\lambda f(\lambda)$.
   The first-order derivative of $F(\lambda)$ on $\lambda$ is given by
    \begin{equation}\label{dF1}
    {F^{(1)}(\lambda)} = f(\lambda)+\lambda f^{(1)}(\lambda)= f(\lambda)G(\lambda).
    \end{equation}
   The introduced auxiliary function $G(\lambda)$ in \eqref{dF1} is defined as $G(\lambda)= 1+
    \frac{\lambda f^{(1)}(\lambda)}{f(\lambda)}$,
    where
    \begin{equation}\label{df1}
      f^{(1)}(\lambda)=\frac{f_1^{(1)}(\lambda)}
      {f_1(\lambda)}-\frac{f_2^{(1)}(\lambda)}
      {f_2(\lambda)},
    \end{equation}
    with $f_1^{(1)}(\lambda)=
        -\frac{{\alpha}(f_1(\lambda)-1)Y}{2(1+\lambda Y)}$ and $f_2^{(1)}(\lambda)=
        -\frac{{\alpha}
        (f_2(\lambda)-1)}{{2}\lambda}$.
    Note that $f(\lambda)$ in \eqref{dF1} is positive, such that the sign of $F^{(1)}(\lambda)$ remains consistent with that of $G(\lambda)$.
First, we investigate the sign of $F^{(1)}(\lambda)$ at the boundaries of $[\lambda_f^L,\infty)$.
A complete expression of $F^{(1)}(\lambda)$ is given by substituting \eqref{df1} into \eqref{dF1}
    \begin{equation}\label{dF11}
  F^{(1)}(\lambda) = \ln\frac{f_1(\lambda)}{f_2(\lambda)}
  +\frac{f_1(\lambda)[f_2(\lambda)-1]
  -\lambda[f_1(\lambda)-f_2(\lambda)]Y}
  {\delta f_1(\lambda)f_2(\lambda)(1+\lambda Y)}.
\end{equation}

Case $\lambda=\lambda_f^L$:
We have $f_1(\lambda^L)=f_2(\lambda^L)$, such that $F^{(1)}(\lambda^L) = \frac{f_1(\lambda^L)[f_1(\lambda^L)-1]}
  {\delta f^2_1(\lambda^L)(1+\lambda^L Y)}>0$.

Case $\lambda\rightarrow\infty$:
We have $\lim_{\lambda\rightarrow\infty}f_1(\lambda)= 1$ and $\lim_{\lambda\rightarrow\infty}f_2(\lambda)= 1$, such that $\lim_{\lambda\rightarrow\infty}
F^{(1)}(\lambda)=
\frac{[f_2(\lambda)-1]
  -\lambda[f_1(\lambda)-f_2(\lambda)]Y}
  {\delta(1+\lambda Y)}$.
  Substituting  $f_1(\lambda)$ and $f_2(\lambda)$ into $\lim_{\lambda\rightarrow\infty}F^{(1)}(\lambda)$ yields
  \begin{align}\label{dF1_INFTY}
    \lim_{\lambda\rightarrow\infty}F^{(1)}(\lambda)
  &=\lim_{\lambda\rightarrow\infty}
  Z\lambda^{-{\alpha}/{2}}\left(1-\frac{X}{Z} Y\left(\frac{\lambda }{1+\lambda Y}\right)^{{\alpha}/{2}+1}\right)\nonumber\\
&  =\lim_{\lambda\rightarrow\infty}
  Z\lambda^{-{\alpha}/{2}}
  \left(1-\frac{X}{Z}Y^{-{\alpha}/{2}} \right)<0,
  \end{align}
where the last inequality holds for $\lambda_f^L=1/\left(\left({X}/{Z}\right)
^{\delta}-Y\right)>0\Rightarrow
\left({X}Y^{-{\alpha}/{2}}\right)/{Z}>1$.
The above two cases also indicate that $G(\lambda^L)>0$ and $\lim_{\lambda\rightarrow\infty}G(\lambda)<0$.

Directly proving the monotonicity of $F^{(1)}(\lambda)$ (or the concavity of $F(\lambda)$) w.r.t. $\lambda$ from \eqref{dF1} is quite difficult.
We observe that, supposing $G(\lambda)$ monotonically decreases with $\lambda$, there obviously exists a unique $\lambda^{\star}$ that makes $F^{(1)}(\lambda)$ first positive and then negative after $\lambda$ exceeds $\lambda^{\star}$. That is, we may prove that $F(\lambda)$ is a first-increasing-then-decreasing function of $\lambda$.
Invoking the definition of the
single-variable quasi-concave function \cite[Sec. 3.4.2]{Boyd2004Convex}, $F(\lambda)$ is actually a quasi-concave function of $\lambda$; the given $\lambda^{\star}$ is the optimal solution that maximizes $F(\lambda)$, which is obtained at
$F^{(1)}(\lambda)=0$.
Based on the above discussion, in what follows
we focus on proving the monotonicity of $G(\lambda)$ w.r.t. $\lambda$.
We first compute the first-order derivative of $G(\lambda)$ on $\lambda$
    \begin{equation}\label{dG1}
 {G^{(1)}(\lambda)} = \frac{f^{(1)}(\lambda)
    	f(\lambda)
    	+\lambda f^{(2)}(\lambda)f(\lambda)
    	-{\lambda}
    	\left(f^{(1)}(\lambda)\right)^2}{f^2(\lambda)}.
    \end{equation}
    Computing ${G^{(1)}(\lambda)}$ requires computing $f^{(2)}(\lambda)$, which can be obtained from \eqref{df1}
    \begin{align}\label{df2}
    {f^{(2)}(\lambda)}& = \frac{f_1^{(2)}(\lambda)f_1(\lambda)
    -(f_1^{(1)}(\lambda))^2}
    {f_1^2(\lambda)}\nonumber\\
&  \qquad  - \frac{f_2^{(2)}(\lambda)f_2(\lambda)
    -(f_2^{(1)}(\lambda))^2}
    {f_2^2(\lambda)},
    \end{align}
    where $f_1^{(2)}(\lambda)=\frac{({\alpha}/{2}+1)
    (f_1(\lambda)-1)Y^2}{\delta(1+\lambda Y)^2}$ and $f_2^{(2)}(\lambda)=\frac{({\alpha}/{2}+1)
   (f_2(\lambda)-1)}{\delta \lambda^2}$
   are the second-order derivatives of $f_1(\lambda)$ and $f_2(\lambda)$, respectively,
substituting which into \eqref{df2} further yields
\begin{align}\label{df2_2}
    f^{(2)}(\lambda) &= \frac{(f_1(\lambda)-1)
    (f_1(\lambda)+{\alpha}/{2})Y^2}
    {\delta(f_1(\lambda))^2(1+\lambda Y)^2}\nonumber\\
    & \qquad
    -\frac{(f_2(\lambda)-1)
    (f_2(\lambda)+{\alpha}/{2})}
    {\delta\lambda^2(f_2(\lambda))^2}.
    \end{align}
Substituting \eqref{df1} and \eqref{df2_2} into \eqref{dG1} yields
\begin{align}\label{dG12}
   & G^{(1)}(\lambda) = \frac{1}{f(\lambda)}
    \Bigg\{-\frac{(f_1(\lambda)-1)Y}
    {\delta f_1(\lambda)(1+\lambda Y)}
    +\frac{(f_2(\lambda)-1)}
    {\delta\lambda f_2(\lambda)}+\nonumber\\
&    \frac{\lambda(f_1(\lambda)-1)
    (f_1(\lambda)+{\alpha}/{2})Y^2}
    {\delta f^2_1(\lambda)(1+\lambda Y)^2}
    -\frac{(f_2(\lambda)-1)
    (f_2(\lambda)+{\alpha}/{2})}{\delta\lambda f_2^2(\lambda)}\nonumber\\
&-\frac{\lambda}{f(\lambda)}
    \left(-\frac{(f_1(\lambda)-1)Y}
    {\delta f_1(\lambda)(1+\lambda Y)}
    +\frac{(f_2(\lambda)-1)}{\delta\lambda f_2(\lambda)}\right)^2\Bigg\}.
\end{align}
Using the inequality $f(\lambda)=\ln\frac{f_1(\lambda)}{f_2(\lambda)}\le
\frac{f_1(\lambda)}{f_2(\lambda)}-1$ and after some algebraic manipulations, we obtain
\begin{align}\label{dG1_up}
     G^{(1)}(\lambda)&\le
-\frac{ f^2_1(\lambda)}{\delta f(\lambda)}[f_2(\lambda)-1]-\frac{\lambda f_1(\lambda)Y }{\delta f(\lambda)}\times
\nonumber\\
    & \Big(f_2(\lambda)[f_1(\lambda)
    -1]+\alpha[f_2(\lambda)-1][f_1(\lambda)
    -f_2(\lambda)]\Big)\nonumber\\
    &-\frac{\lambda^2f_2(\lambda)Y^2}{\delta f(\lambda)}
    [f_2(\lambda)-1]
    [f_1(\lambda)-f_2(\lambda)]^2.
\end{align}
Given that $f_1(\lambda)>f_2(\lambda)>1$, all the coefficients of $Y^i$ for $i=0,1,2$ in the right-hand side (RHS) of \eqref{dG1_up} are negative, such that $G^{(1)}(\lambda)<0$.
This means $G(\lambda)$ is a monotonically decreasing function of $\lambda$ in the range $[\lambda_f^L,\infty)$.
By now, we have completed the proof.

\subsection{Proof of Theorem \ref{ptnt_bound_theorem}}
\label{appendix_ptnt_bound_theorem}
The proof is similar to Appendix \ref{appendix_pt_bound_theorem}; the only difference lies in computing $\mathcal{L}_{I_f}(s)$, which is obtained from
\eqref{sir_fd_nt} and \eqref{la_fd_low}
    \begin{align}\label{la_fd_nt_low}
 &\mathcal{L}_{I_f}(s)
= \mathbb{E}_{\hat \Phi_f}\left[
\prod_{\hat z\in\hat \Phi_f\setminus\hat  o}e^{-s\left(\frac{P_f |\tilde{\bm w}_{f} \bm f_{ \hat z o}|^2}{D_{ \hat z o}^{\alpha}} +\frac{P_t}{N_j}\frac{\|\tilde{\bm w}_{f} \bm F_{z o}\tilde{\bm F}_{z}\|^2}{D_{z o}^{\alpha}}\right)}\right]\ge\nonumber\\
 &
\mathbb{E}_{\Phi_f}\left[
\prod_{ z\in\Phi_f\setminus o}e^{-\frac{sP_f |\tilde{\bm w}_{f} \bm f_{ \hat z o}|^2}{D_{ \hat z o}^{\alpha}}}\right]
\mathbb{E}_{\hat\Phi_f}\left[
\prod_{\hat z\in\hat \Phi_f\setminus \hat o}e^{-\frac{sP_t}{N_j}
	\frac{\|\tilde{\bm w}_{f} \bm F_{zo}\tilde{\bm F}_{z}\|^2}{D_{zo}^{\alpha}}}\right]\nonumber\\
& \stackrel{\mathrm{(k)}} =
e^{-C_{\alpha,2}\lambda_fD_f^2}
e^{-C_{\alpha,N_j+1}
	\left({P_{tf}}/{N_j}\right)^{\delta}
	\lambda_fD_f^2},
\end{align}
where (k) holds for \cite[(8)]{Haenggi2009Stochastic} combined with $\|\tilde{\bm w}_{f} \bm F_{zo}\tilde{\bm F}_{z}\|^2\sim\Gamma(N_j,1)$.
Substituting \eqref{la_hd} and \eqref{la_fd_nt_low} into \eqref{pt_exact} and invoking \cite[Theorem 1]{Hunter08Transmission}, we complete the proof.

\subsection{Proof of Theorem \ref{psnt_approx_theorem}}
\label{appendix_psnt_approx_theorem}
    Following \eqref{ps_app1}, we first compute $\mathbb{P}
    \left\{\widetilde{\textsf{SIR}}_{e}
    \ge\beta_{e}\right\}$.
    Recalling \eqref{sir_e_nt}, each term in $\tilde{\bm R}_{e}$, e.g., $\bm G_{ze}
\tilde{\bm F}_{z}\tilde{\bm F}_{z}^H\bm G_{ze}^H$, can be regarded as a superposition of single-stream signals with $N_j$ co-located interferers.
Denote the $n$-th column of $\bm G_{ze}\tilde{\bm F}_{z}$ by $\tilde{\bm g}_{z e,n}$, then
$\bm R_{e,N_t,N_j}$ can be reformed as
\begin{equation}\label{R_e_nt}
    \tilde{\bm R}_{e} =
    \frac{P_t}{N_j}\sum_{n=1}^{N_j}\frac{\tilde{\bm g}_{o e,n} \tilde{\bm g}_{ o e,n}^H}{ D_{ o e}^{-\alpha}}
    +\sum_{z\in\Phi_f\setminus o}\frac{P_t}{N_j}\sum_{n=1}^{N_j} \frac{\tilde{\bm g}_{z e,n} \tilde{\bm g}_{z e,n}^H}{D_{z e}^{-\alpha}}.
\end{equation}

Define $r\triangleq D_{oe}$ and $z\triangleq {r^{\alpha}\beta_e}{P_f}$.
 $\mathbb{P}\left\{\textsf{SIR}_{e,N_t,N_j}
    \ge\beta_{e}\right\}$ is obtained by invoking \cite[(11)]{Gao1998Theoretical}, i.e.,
\begin{equation}\label{pro_sir_nt}
  \mathbb{P}
    \left\{\widetilde{\textsf{SIR}}_{e}
    \ge\beta_{e}\right\}
    =    \mathbb{E}_{\Phi_f}\left[
    \frac{1}{W_{N_j}}\sum_{n=0}^{N_e-1}y_nz^n\right],
\end{equation}
where $W_{N_j}=\left(1+\frac{P_t}{N_j}
r^{-\alpha}z\right)^{N_j}\prod_{z\in \Phi_f\setminus o}\left(1+\frac{P_t}{N_j}D_{ze}^{-\alpha}z\right)^{N_j}$ and $y_n$ is the coefficient of $z^n$ in the polynomial expansion of $W_{N_j}$.
Define $\tilde{A}\triangleq C_{\alpha,N_j+1}\lambda_f
    \left({P_{tf}\beta_e}/{N_j}\right)^{\delta}$.
    Invoking \cite[Theorem 1]{Louie2011Spatial} yields
    \begin{align}\label{pro_exp}
    \mathbb{P}\left\{\widetilde{\textsf{SIR}}_{e}
    \ge\beta_{e}\right\} &=
    \sum_{n=0}^{N_e-1}
        \sum_{i=0}^{\min(n,N_j)}\binom{N_j}{i}
        {\left(\frac{P_{tf}\beta_e}{N_j}\right)^{i}}
        \times\nonumber\\
        &\sum_{j=1}^{|\xi_{n-i}|}\frac{\Xi_{j,n-i}
{(-\tilde{A}r^2)^{|\xi_{j,n-i}|}e^{- \tilde{A}r^2}}}
{\left(1+{P_{tf}\beta_e}/{(N_j)}\right)^{N_j}},
    \end{align}
Substituting \eqref{pro_exp} into \eqref{ps_app1} and using $D_f\rightarrow 0$ and formula \cite[(3.326.2)]{Gradshteyn2007Table}, we complete the proof.

\end{document}